\newif\ifconference
\newif\ifanonymous
\conferencetrue
\anonymousfalse

\documentclass[11pt,letterpaper,final]{article}

\title{Heaps and Their Working Sets}

\author{
	\ifanonymous
	Anonymous Authors
	\else
	Bernhard Haeupler\thanks{INSAIT, Sofia University ``St.~Kliment Ohridski'' \& ETH Zurich, \texttt{bernhard.haeupler@insait.ai}}
	\and
	Richard Hladík\thanks{ETH Zurich, \texttt{rihl@rihl.cz}}
	\and
	Václav Rozhoň\thanks{Charles University, \texttt{vaclavrozhon@gmail.com}}
	\and
	Robert E.~Tarjan\thanks{Princeton University, \texttt{ret@cs.princeton.edu}}
	\fi
}

\ifanonymous

\else

\fi

\date{}

\usepackage[margin=1in]{geometry}
\usepackage[dvipsnames,svgnames]{xcolor}
\usepackage{algorithm}
\usepackage{algpseudocode}
\usepackage{xspace}
\usepackage[textwidth=20mm,obeyFinal,colorinlistoftodos]{todonotes}
\usepackage{amsmath}        
\usepackage{amsfonts}       
\usepackage{amsthm}         
\usepackage{thmtools}
\usepackage{thm-restate}
\usepackage{amssymb}
\usepackage{bbding}         
\usepackage{bm}             
\usepackage{dsfont}
\usepackage{textcomp}
\usepackage{graphicx}       
\usepackage{paralist}
\usepackage{mathtools}
\usepackage{multirow}
\usepackage{booktabs}
\usepackage{adjustbox}
\newcommand{\mytodo}[2]{\todo[size=\tiny, color=#1!50!white]{#2}\xspace}

\newcommand{\vasek}[1]{\mytodo{purple}{Vasek: #1}}

\setuptodonotes{size=\tiny}
\usepackage[backend=biber,style=alphabetic,sorting=nty,maxnames=10]{biblatex}         
\usepackage{hyperref}
\hypersetup{
    unicode=false,          
    colorlinks=true,        
    linkcolor=red,          
    citecolor=DarkGreen,        
    filecolor=magenta,      
    urlcolor=cyan           
}
\usepackage{enumitem}
\usepackage{graphicx} 
\usepackage{appendix}
\usepackage[capitalize, nameinlink,noabbrev]{cleveref}
\usepackage[left,mathlines]{lineno}

\usepackage{setspace}
\usetikzlibrary{arrows.meta}

\renewcommand\O{\mathcal{O}}

\theoremstyle{plain}
\newtheorem{theorem}{Theorem}[section]
\newtheorem{lemma}[theorem]{Lemma}
\newtheorem{meta-theorem}[theorem]{Meta-Theorem}
\newtheorem{claim}[theorem]{Claim}

\theoremstyle{definition}
\newtheorem{definition}[theorem]{Definition}

\theoremstyle{remark}
\newtheorem{remark}[theorem]{Remark}

\crefname{theorem}{Theorem}{Theorems}
\crefname{proposition}{Proposition}{Propositions}
\crefname{observation}{Observation}{Observations}
\crefname{lemma}{Lemma}{Lemmas}
\crefname{claim}{Claim}{Claims}
\crefname{problem}{Problem}{Problems}
\crefname{conjecture}{Conjecture}{Conjectures}
\crefname{question}{Question}{Questions}
\crefname{example}{Example}{Examples}
\crefname{fact}{Fact}{Facts}    
\crefname{invariant}{Invariant}{Invariants}
\crefname{rule}{Rule}{Rules}

\let\cref=\Cref
\let\citet\textcite

\newcommand\N{\mathbb{N}}

\DeclareMathOperator*{\argmax}{arg\,max}

\graphicspath{{fig/}}

\newcommand{\OO}{\mathrm{O}}
\let\O\OO

\def\defop#1#2{\expandafter\def\csname#1\endcsname{\textsc{#1}}}
\defop{Insert}{insert}
\defop{DecreaseKey}{decrease-key}
\defop{IncreaseKey}{increase-key}
\defop{FindMin}{find-min}
\defop{ExtractMin}{delete-min}
\defop{Delete}{delete}
\defop{PushLeft}{push-left}
\defop{PushRight}{push-right}
\defop{Next}{next}
\defop{Prev}{prev}
\defop{First}{first}
\defop{Last}{last}
\defop{HeapPtr}{heap-ptr}
\defop{LeftmostBundle}{leftmost-bundle}
\defop{RightmostBundle}{rightmost-bundle}
\defop{BatchPushLeft}{batch-push-left}
\defop{BatchPushRight}{batch-push-right}
\defop{BatchPopLeft}{batch-pop-left}
\defop{BatchPopRight}{batch-pop-right}
\defop{Remove}{remove}
\defop{LowerBound}{lower-bound}
\defop{Create}{create}
\defop{FixTooBig}{fix-too-big}
\defop{FixTooSmall}{fix-too-small}
\defop{Merge}{merge}
\defop{Touch}{touch}
\let\Decreasekey\DecreaseKey

\let\Deletemin\ExtractMin
\newcommand\Strictw{\tilde W}
\newcommand\strictw{\tilde w}

\def\ref#1{\textcolor{red}{[\textbackslash ref is disabled. Please use \textbackslash cref instead. Use \textbackslash cref\{sec:foo,sec:bar\} to reference two things at once.]}}

\def\ourabstract{%
We construct a heap with strong beyond-worst-case performance guarantees and explore the analysis of such heaps.
\texorpdfstring{\par}{}%
First, we unify existing notions of the working-set bound for heaps by proving that essentially all of them are equivalent -- with the notable exception of the so-called stack-like bound, which is strictly stronger. This equivalence simplifies the theoretical landscape and extends the range of applications of heaps with working-set bounds.
\texorpdfstring{\par}{}%
Second, we present the first heap implementation that has the amortized stack-like bound and supports $\O(1)$-time decrease-key and $o(\log^*n)$-time insert.
}

\begin{document}

\maketitle

\begin{abstract}
\ourabstract
\end{abstract}

\newpage
\section{Introduction}
\label{s:intro}

This paper presents a new heap with strong beyond-worst-case guarantees. We also simplify the understanding of working-set bounds, a family of seemingly different beyond-worst-case guarantees, by proving that almost all known working-set bounds are equivalent.

A heap is a data structure that supports the following operations: $\Insert(x)$ that inserts a new item, $\Decreasekey(x, x')$ that decreases the value of an item in the data structure to a new value, and $\Deletemin$ that deletes and returns the currently smallest item in the data structure. Heaps are among the most basic data structures used as algorithmic building blocks, in applications like sorting, Dijkstra's algorithm and minimum spanning trees.

This paper explores their beyond-worst-case properties in the comparison model. More concretely, while it is well-understood that classical implementations of heaps cannot be improved in the worst case, we explore heaps with much stronger guarantees on natural classes of inputs. One important class are the \emph{working-set properties}, a family of temporal-locality-based guarantees that have found applications in fundamental problems, including optimal merging \cite{iacono}, Dijkstra's algorithm \cite{dijkstra} and sorting problems \cite{supi}. 

Our contribution is twofold. First, we advance the theory behind the working-set properties of heaps, showing that, surprisingly, essentially all natural definitions of the working-set property are equivalent in the amortized sense. The only exception is the \emph{stack-like property}, the strongest-known working-set style property. Second, we construct a heap achieving the stack-like property while also supporting all operations in optimal time, except for a negligible $\log^*$-type additive term.


\subsection{Beyond-worst-case algorithms and data structures}

After decades of continuous progress, we now understand that many known data structures are the best possible in the comparison model. 
For example, consider the ordered dictionary problem that asks for a data structure supporting operations $\Insert(x)$, $\Remove(x)$, and $\LowerBound(x)$ that finds the smallest element $y \ge x$. This problem is solved by classical binary search trees like AVL trees \cite{avl} or red-black trees \cite{redblack}: Each operation costs $\O(\log n)$ time. At the same time, we know that in the comparison model, at least one of the three operations has to take time $\Omega(\log n)$, since the ordered dictionary problem can be used to sort. 

While we cannot improve the worst-case complexities of those fundamental data structures, we can try to create data structures that are much faster than worst-case for some important classes of inputs. For example, a famous dynamic-optimality conjecture posits that splay trees~\cite{splay} -- a concrete implementation of binary search trees -- perform a sequence of access queries asymptotically as fast as \emph{any} other binary search tree. If true, splay trees are a much more satisfying solution to the ordered dictionary problem than the classic worst-case optimal solutions \cite{avl,redblack}. 

An important theoretical reason to study beyond-worst-case data structures is that questions about beyond-worst-case algorithms often lead to questions about beyond-worst-case data structures. 
To give an example, recently, \citet{dijkstra} showed how replacing Fibonacci heaps with sufficiently fast beyond-worst-case heaps in Dijkstra's algorithm turns a worst-case-optimal algorithm into a so-called \emph{universally optimal} one. Such an algorithm runs for each graph in time that is the best possible for that graph.\footnote{The problem solved by Dijkstra's algorithm is to compute the distance from the source vertex to all vertices, and order the vertices by this distance.}
Importantly, the result of \citet{dijkstra} works with any heap, provided that 1) it has the so-called \emph{working-set property}, defined shortly, and 2) it still has the Fibonacci-heap guarantee of $\O(1)$-time $\Insert$ and $\DecreaseKey$.

Similarly, \citet{supi} recently resolved the long-standing problem of sorting under partial information with asymptotically optimal time and query complexity\footnote{In this problem, we are given a set of $n$ elements and $m$ answers to some pairwise comparisons of those elements. The task is to sort the elements using the asymptotically best possible time and number of additional comparisons.} with a natural \emph{topological heapsort} algorithm whose strength comes from using a heap with the working-set property.

The above discussion suggests that the design of beyond-worst-case algorithms may frequently reduce to the design of beyond-worst-case data structures, and heaps with the working-set property in particular. We discuss this property next. 

\subsection{Heaps and their working sets}

One natural class of beyond-worst-case properties are those that make use of temporal locality. The intuition is that if an item $x$ was accessed recently, a good heuristic is to assume it will be accessed again soon. Thus, if accessing such items is cheaper, then we should perform better in applications with good temporal locality.

As an example, say we have a magical way of telling if $x$ is the heap's minimum. Then we can store items on a stack in insertion order and perform $\Deletemin$ by descending down the stack until we find the minimum $x$. This takes $\Theta(s_x)$ time, where $s_x$ is the depth of $x$. In applications with good temporal locality, this beats standard heaps. However, in general, each $\Deletemin$ can take up to $\Theta(n)$ time. With wishful thinking, we could hope for a better, $\O(\log s_x)$-time guarantee, as that would nicely interpolate to $\O(\log n)$ even in the worst case.
This is exactly the \emph{stack-like property:}

\begin{definition}[Stack-like property \cite{elmasry}]
	Define the \emph{stack size} of $x$ as the number $s_x$ of items in the heap that are no older than (inserted no earlier than) $x$, including $x$, at the moment $x$ is deleted.
	A heap has the \emph{stack-like property} if the amortized cost of extracting an item $x$ is $\O(\log s_x)$.\footnote{Note that for $s_x = 1$, this requires $\O(0)$-time $\ExtractMin$. This could be problematic, but all our beyond-worst-case properties are amortized, and thus each $\ExtractMin$ can charge $\O(1)$ to the corresponding $\Insert$. This way, we can avoid writing $\O(1 + \dots)$ every time.}
\end{definition}

Notably, if we happened to use a stack-like heap as a stack and always delete the most recent item, we would get constant-time deletion. Paired with $\O(1)$-time $\Insert$ and $\DecreaseKey$, this leads to a data structure that can at the same time be used as a fast heap and a stack, and gracefully interpolates between both regimes.
There are known stack-like heaps \cite{elmasry,elmasry_farzan_iacono,unified-bound-heaps}, but none supports $\DecreaseKey$ faster than via deleting and reinserting (in \cite{elmasry_farzan_iacono}, it is mentioned as an open problem).
We construct a stack-like heap with $\O(1)$-time $\DecreaseKey$, with only an additive $o(\log^*n)$ slowdown per $\Insert$.

\paragraph{Working-set property}

The stack-like property is very strong and thus hard to achieve.
A weaker variant of the idea that more recent items should be cheaper to access is the so-called \emph{working-set property}. It requires that the deletion cost is logarithmic in the time since insertion:

\begin{definition}[Working-set property]
	A heap has the \emph{working-set property} if the amortized cost of extracting an item $x$ is $\O(\log (t'_x - t_x))$ where $t_x,t'_x$ are the times of insertion and extraction of $x$, respectively. 
\end{definition}

It can immediately be seen that the stack-like property is stronger than the working-set property, since $s_x \le t'_x - t_x$ for all $x$. It is actually even asymptotically stronger: if we insert numbers $n$ to $1$ sequentially and then perform $n$ $\Deletemin$s, this has $\O(n)$ cost with a stack-like heap (as we always delete the most recent item remaining), but the working-set property only guarantees $\O(\log 1 + \log 3 + \ldots + \log (2n - 1)) = \O(n \log n)$.


\paragraph{Other working-set properties}

In the original work that established the terminology, \citet{iacono} defines a third property slightly stronger than the standard working-set property that we call for that reason the \emph{strong working-set property} (\cref{def:strong_working_set}). The two aforementioned results \cite{dijkstra,supi} used strong working-set heaps in their original preprints \cite{dijkstra_old,supi_old} as their proof techniques required them. \citet{dijkstra-simpler,supi-simpler} later showed a simpler analysis for which regular working-set heaps are sufficient.

Similarly, \citet{supi} define the even stronger \emph{strict working-set property} (\cref{def:strict_working_set}) and construct a heap achieving it to get one of their results.

Working-set properties have proven extremely useful, with each variant finding its own applications. The differences between variants seem to be important, and previous results go out of their way to construct heaps with properties strong enough for their applications \cite{dijkstra,supi} or to prove that weaker heaps suffice \cite{dijkstra-simpler,supi-simpler}. Given that we additionally identify three more natural working-set variants, some of which are also used in the literature, the resulting landscape is somewhat fragmented. Our contribution is to significantly simplify it.

\subsection{Our results}
\label{sec:results}

This paper shows new results related to the stack-like property and working-set properties.

\subsubsection{Equivalence between various working-set properties}

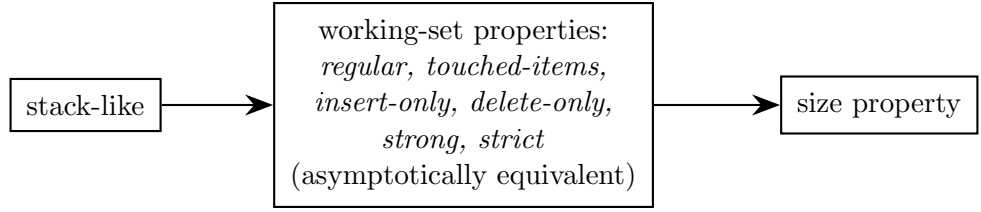
\begin{figure}
    \centering
\begin{tikzpicture}[
  every node/.style={draw, rectangle, align=center, inner sep=6pt},
  ->, >={Stealth[scale=1.6]}, thick,
]
  \node (size)  at (-2, 0)      {size property};
	\node (ws)    at (-7.5, 0)  {working-set properties:\\\textit{regular, touched-items,}\\\textit{insert-only, delete-only,}\\\textit{strong, strict}\\(asymptotically equivalent)};
  \node (stack) at (-12.5, 0)   {stack-like};

  \draw (ws)   -- (size);
  \draw (stack)      -- (ws);
\end{tikzpicture}
	\caption{Asymptotic relationships between beyond-worst-case properties discussed in \cref{s:weakstrong}.}
    \label{fig:diagram}
\end{figure}

Our first conceptual contribution is showing that essentially all working-set properties found in the literature, as well as some other natural definitions of the working-set property, are equivalent.

\begin{restatable}{theorem}{allpropertiessame}
\label{thm:all_same}
	All of the following properties are equivalent, up to constant factors and in the amortized sense:
	\begin{enumerate}[itemsep=0pt]
        \item working-set property (\cref{def:working_set}),
		\item touched-items working-set property (\cref{def:insert_delete_working_set}),
        \item insert-only working-set property (\cref{def:insert_delete_working_set}),
        \item delete-only working-set property (\cref{def:insert_delete_working_set}),
        \item strong working-set property (\cref{def:strong_working_set}),
        \item strict working-set property (\cref{def:strict_working_set}).
    \end{enumerate}
    That is, any heap satisfying one of the properties also satisfies all the others, up to constant factors, in the amortized sense, and while preserving the cost of all other operations.

\end{restatable}

We believe that this is a surprising equivalence since some of the definitions seem to be very different, and are treated as such in the literature, as previously discussed.

To our knowledge, \cref{thm:all_same} covers all working-set-style heap properties found in the literature, except for the stack-like property, which is strictly stronger:

\begin{restatable}{lemma}{stacklikestronger}
\label{lem:stacklike_stronger}
    The stack-like property is strictly stronger than all of the working-set properties listed in \cref{thm:all_same}, up to constant factors and in the amortized sense. That is:
    \begin{enumerate}[itemsep=0pt]
    	\item any heap satisfying the stack-like property also satisfies all of the listed working-set properties, up to constant factors, and in the amortized sense; and
    	\item there is an infinite family of instances on which the total cost guarantees given by the stack-like property are asymptotically stronger than those given by the working-set property.
    \end{enumerate}
\end{restatable}

The resulting relationships between heap beyond-worst-case properties are shown in \cref{fig:diagram}.
We note that both \cref{thm:all_same} and \cref{lem:stacklike_stronger} work black-box for any heap/priority-queue implementation.

Together, these results significantly simplify the landscape and applications of working-set properties. Whenever an application calls for a working-set property, then either it requires the stack-like property, or we can just pick our favorite heap implementation that has \emph{some} working-set-style property, and chances are that it will fall into the cluster covered by \cref{thm:all_same}.

As an example, consider again the work of \citet{supi}, who construct a heap with the strict working-set property to achieve one of their results. Using \cref{thm:all_same}, we immediately get the same result just by using any working-set heap.


\subsubsection{A new heap with stack-like property and constant $\Decreasekey$}

Second, we construct a heap that has the stack-like property, while it also supports the \Decreasekey{} operation in $\O(1)$ amortized time and $\Insert$ in $\O(f(n))$ amortized time for a very slowly growing $f$.

\begin{restatable}[Informal version of \cref{thm:heap_main}]{theorem}{ackerheap}
\label{thm:heap_main_informal}
	For any $k \in \N$, there exists a heap that supports operations \Insert, \Decreasekey, and \Deletemin{} in the following amortized time: 
	\begin{enumerate}[itemsep=0pt]
		\item \Insert{} in time $\O(\log^{*(k)}n)$,
            \item \Decreasekey{} in time $\O(1)$,
		\item \Deletemin{} in time $\O(\log s_x)$.
    \end{enumerate}
	Here $\log^{*(k)}$ is the $k$-fold iterated logarithm, that is, the $k$-th function in the sequence $\log^*$, $(\log^*)^*$, $((\log^*)^*)^*$, \dots{}, defined formally in \cref{def:k-fold}.
\end{restatable}


\cref{thm:heap_main_informal} is proven by a careful recursive construction, starting from the Fibonacci heap and gradually bringing the additive overhead of $\Insert$ down at the expense of making other operations $\O(1)$-times slower. The overall challenge comes from the fact that our heap has to support fast \Decreasekey{} operations. At the same time, we have to keep track of the insertion order of the items, to support stack-like \Deletemin. 

We note that up to the $\O(\log^{*(k)}n)$ additive overhead, and the fact that bounds are amortized, our construction improves upon the heaps known in the literature. First, our data structure matches the fast decrease-key supported by Fibonacci heaps \cite{fibonacci}, hollow heaps \cite{hollow}, Brodal heaps \cite{brodal} or rank-pairing heaps \cite{rank_pairing}. This is not known to be the case for many beyond-worst-case heaps: funnel heaps \cite{funnel,elmasry}, pairing heaps \cite{pairing,iacono} and the heaps by \citet{elmasry,elmasry_farzan_iacono,unified-bound-heaps}, most of which are not known to support $\DecreaseKey$ better than by deleting and re-inserting. On the other hand, the only beyond-worst-case heaps that are known to support fast $\DecreaseKey$ \cite{dijkstra,dijkstra-simpler,dijkstra-pm} are only known to have the working-set, but not the stack-like property.


\subsubsection{Simple heaps with the working-set and stack-like property}

We emphasize that $\O(1)$-time $\DecreaseKey$ is a crucial part of our stack-like heap result.
Generally, requiring fast $\DecreaseKey$ seems to make designing beyond-working-set heaps more challenging: while working-set and stack-like heaps without fast $\DecreaseKey$ have been known since the early 2000s, the result of \citet{dijkstra} was the first one to achieve a working-set heap with $\O(1)$-time $\DecreaseKey$.

One possible reason is that fast $\DecreaseKey$ is what sets heaps apart from BSTs. In fact, some straightforward constructions with BSTs directly lead to heaps with the working-set and stack-like properties. For completeness, we sketch some of them in \cref{sec:simple-heaps}.

\subsection{More related work}
\label{sec:related_work}

There is a long line of work on beyond-worst-case heaps; besides the results already mentioned, we highlight the paper of \citet{munro2019dynamic} that proves that the heap analog of the dynamic-optimality conjecture \cite{splay} fails for a wide family of heaps known as the tournament heaps, and the paper by \citet{kozma_saranurak} who show a close correspondence between beyond-worst-case heaps and binary search trees. We refer the reader to the surveys  \cite{risaSurveyPic,bernhardsSurvey} about other input-sensitive time bounds for heaps and data structures in general. 
\citet{unified-bound-heaps} mentioned above constructs a heap that achieves a so-called \emph{unified bound}, which goes beyond our working-set classification and implies the stack-like property. However, his heap does not support fast $\DecreaseKey$.

There is even a longer line of research related to the beyond-worst-case binary search trees \cite{tango-trees,unified-bound}, motivated by the dynamic optimality conjecture by \citet{splay}. The working-set property is only one of a number of beyond-worst-case properties attained by splay trees; the others include the static finger property \cite{splay}, the dynamic finger property \cite{dynamic-finger1,dynamic-finger2}, static optimality \cite{splay} and others.

\section{Technical Overview}
\label{sec:overview}
\label{subsec:sketch_heap}

Here we provide an informal overview of our most technical contribution, a stack-like heap with fast $\DecreaseKey$. Our other main result, about the equivalence between working-set properties, is treated in a self-contained way in \cref{s:weakstrong}.

We repeat \cref{thm:heap_main_informal} here for convenience.

\ackerheap*

\paragraph{Construction with $\log\log n$ instead of $\log^{*(k)}n$}
We first sketch how one can prove \cref{thm:heap_main_informal} with a somewhat worse loss of $\log\log n$ in the complexity of \Insert. The construction is also shown in \cref{fig:heap}. 

\begin{figure}
    \centering
    \includegraphics[width=\linewidth]{img/heap.pdf}
    \caption{Diagram of our basic heap construction: 
The data structure consists of warehouses of doubly-exponentially-growing sizes. The construction aims to preserve the linked-list order of the items, newly inserted items are thus inserted in the leftmost, smallest, warehouse. \\
We group the items in each warehouse in bundles of size that is roughly the logarithm of the size of this warehouse (and also the neighboring, smaller, warehouse). Each warehouse then stores the minima of the bundles in a Fibonacci heap. \\
In the improved construction, we use a similar construction recursively, replacing the Fibonacci heap by an already constructed heap. }
    \label{fig:heap}
\end{figure}

To do so, we use a standard idea (see e.g. \cite{unified-bound,dijkstra}) of storing the data in heaps of sizes that grow doubly exponentially. We maintain $\O(\log \log n)$ data structures that we call warehouses, where each warehouse is responsible for storing a sequence of items that are contiguous in their insertion order. The warehouses $W_1, W_2, \dots, W_t$ are such that the $i$-th warehouse stores roughly $w_i =  2^{2^i}$ items. That is, the sizes of the warehouses grow doubly exponentially. In fact, we keep it as an invariant that each $i$-th warehouse contains between $w_i/2$ and $w_i$ items, with the exception of the last, largest warehouse that can be smaller.

Each warehouse is a Fibonacci heap. However, the heap does not store its items directly, as is the case e.g. in the data structure from \cite{dijkstra}. Instead, adjacent items in the insertion order are grouped into \emph{bundles}. Each bundle is essentially a small linked list of items represented by their minimum value in the warehouse Fibonacci heap. The bundles themselves are also organized in an auxiliary linked list maintained by the warehouse. In the warehouse $W_i$, each bundle contains roughly $b_i = \log(w_i)$ items.

We implement the three basic operations \DecreaseKey{}, \Insert, \ExtractMin{} as follows. We note that these operations can be implemented in the running times $\O(1), \O(\log \log n), \O(\log s_x + \log \log n)$. Furthermore, the additive $\log \log n$ term of $\ExtractMin$ can be charged to $\Insert$, and we get the stated complexities.
\begin{enumerate}
    \item $\DecreaseKey(x, x')$. We check whether the item becomes the new minimum in its bundle. If it does, we run $\DecreaseKey$ in the appropriate Fibonacci heap. This operation is thus implemented in $\O(1)$ amortized time. 
    \item $\Insert(x)$. We insert the item into the leftmost warehouse with bundles of size one. We also add the bundle to the warehouse's linked list, at the leftmost end.  This can be done in $\O(1)$ time. 
    Then, we run the \FixTooBig{} operation (see below).
    \item $\ExtractMin$. First, we find the warehouse containing the minimum in $\O(t) = \O(\log \log n)$ time. Next, we run $\ExtractMin$ from that warehouse in time $\O(\log|W_i|) = \O(\log w_i)$. 
     After the extraction, some bookkeeping is in order: We need to delete the item from its bundle and find its new minimum; this again takes $\O(b_i) = \O(\log w_i)$ steps. Afterward, we need to reinsert the bundle back into the warehouse, which takes constant time. The total time spent is thus $\O(\log w_i)$. 

     Finally, we call $\FixTooSmall{}$.
     
    Importantly, since the $(i-1)$-th warehouse contains at least $2^{2^{i-1}} / 2$ items that are all younger than $x$, we have for the extracted item $x$ that $s_x \ge 2^{2^{i-1}} / 2$. Combining this with the bound $w_i \le 2^{2^i}$, we notice that 
    \[
	    \log w_i \le \log 2^{2^i} \le 2 \cdot \log 2^{2^{i-1}} \le 2 \cdot \log (2s_x). 
    \]
    That is, the overall cost of extraction $\O(\log w_i)$ can be upper bounded by the cost $\O(\log s_x)$ required by the stack-like property. 
\end{enumerate}

Next, let us sketch the idea behind the remaining bookkeeping operations for maintaining correct warehouse sizes; the real implementation is slightly more technical.\footnote{For one, some bundles may contain fewer than $b_i$ items, and we might thus need to move many bundles between $W_i$ and $W_{i+1}$, each of which still costs us $\log w_i$. Fortunately, this amortizes out with the right implementation.} Both operations are run for $i = 1, \ldots, t$. After the $i$-th call, the first $i$ warehouses again satisfy the size constraints.

\begin{enumerate}
	\item $\FixTooBig{}$. If the $i$-th warehouse $W_i$ contains more than $w_i$ items, use the linked list of $W_i$ to find its rightmost bundle and move it to $W_{i+1}$. This requires running an \ExtractMin{} operation on $W_i$ that costs us $\O(\log |W_i|) = \O(\log w_i) = \O(b_i)$ steps. We therefore spend $\O(1)$ time per item moved from $W_i$ to $W_{i+1}$.
    \item $\FixTooSmall{}$. If the $i$-th warehouse is too small, we make the neighboring, larger warehouse $W_{i+1}$ send it back some items. As in the $\FixTooBig$ case, this requires one \ExtractMin{} operation that however amortizes to $\O(1)$ time per item moved. 
\end{enumerate}

To understand the overall time complexity of all the calls to $\FixTooBig$ and $\FixTooSmall$ operations, we use in \cref{sec:stacklike_heap} a standard amortization argument that essentially tells us that in the amortized sense, $1+o(1)$ items move between $W_i$ and $W_{i+1}$ per every $\Insert$ and $\Delete$. As each move costs us $\O(1)$ time per item and there are $t = \O(\log \log n)$ levels, we pay a total of $\O(\log\log n)$ steps per every $\Insert$ and $\Delete$.

\paragraph{Turning $\log\log n$ to $\log^* n, \log^{**} n, \dots$}
To improve the error term from $\log\log n$ to $\log^{*(k)}(n)$, we iteratively improve our data structure. In the $i$-th step of our construction, we rely on the previous construction, but instead of using Fibonacci heaps as the basic data structure to support the warehouse operations, we simply use the data structure from step $i-1$. 

To give an example of how the construction can be improved in the first recursive step, notice that in the construction above, we needed to have a relation $|W_{i-1}| \approx |\sqrt{W_i}|$. This was to make sure that if we run an $\ExtractMin$ of cost $\O(\log |W_i|)$ in $W_i$, we can use the observation that $s_x \ge |W_{i-1}|$ to upper bound this cost by $\O(\log s_x)$. 

Now, let us replace the Fibonacci heap with a heap that supports $\ExtractMin$ in time $\O(\log s_x + f(n))$ for a very slowly growing function $f$; for example, the heap constructed above has $f(n) = \O(\log\log n)$. Now, a careful analysis of the above construction shows that the optimal choice of parameters is $w_{i} = 2^{f(w_{i+1})}$ and $b_i = f(w_i)$. That is, for slow-growing $f$, we can let the warehouse sizes grow quicker and let the bundles be smaller portions of warehouses. 

In particular, the equation $w_{i} = 2^{f(w_{i+1})}$, for a function that grows as $f(n) = \O(\log \log n)$ or slower, implies that $w_i < f(w_{i+2})$ and, thus, we conclude that we only need $\O(f^*(n))$
warehouses in our data structure. This means that we can replace the \ExtractMin{} cost of $\O(\log s_x + f(n))$ of the old data structure with the new, improved cost of $\O(\log s_x + f^*(n))$. This is at the cost of somewhat blowing up the constant factors in the complexities of our operations. 

Iterating the same procedure a constant number of times yields \cref{thm:heap_main_informal}. Some additional care is needed: For example, the recursive approach requires that we build a heap with the so-called \emph{deque-like property}, even if we only want to construct a heap with the stack-like property. The deque-like property requires that $\ExtractMin$ is cheap not only for very recent items, but also for very old items. This is because internal functions $\FixTooBig$ and $\FixTooSmall$ are moving items both to lower and higher warehouses; this requires fast \ExtractMin{} operations for bundles both at the low and high end of the insertion order within the warehouse.

\section{A stack-like heap with fast insert and decrease-key}
\label{sec:stacklike_heap}

In this section, we construct a heap that is stack-like and at the same time supports $\Insert$ and $\DecreaseKey$ in optimal time, up to a very small additive loss, thereby proving \cref{thm:heap_main_informal} (restated more formally below as \cref{thm:heap_main}).
We do so by progressively constructing stack-like heaps with smaller and smaller additive loss.
For a high-level sketch of the construction, see \cref{subsec:sketch_heap}.

Our heap will actually be deque-like, which is a strictly stronger property. We show in \cref{lem:stack_implies_deque} that any stack-like heap can be made deque-like with only constant-factor losses to all its operations.

We start by formally defining the family of slow-growing functions that appear in our time complexities.

\begin{definition}[$k$-fold iterated logarithm]
	\label{def:k-fold}
	Define $\log^{*(0)}(n)$ as $\log(n)$.\footnote{All logarithms in this paper are binary and we use the convention that $\log x = 0$ for $x \le 1$.} Now, define inductively $\log^{*(k + 1)}(n)$ as the smallest integer $i$ such that 
\[
	\underbrace{\log^{*(k)}(\log^{*(k)}(\dots \log^{*(k)}(}_{\text{$i$ times}}n)\dots)) \le 1.
\]
\end{definition}

\paragraph{Deque-like interface}

It helps to think of our data structure not primarily as a heap, but rather as a doubly-ended queue (deque) that happens to also support heap operations. In particular, our data structure maintains an implicit doubly linked list of all its items. Maintaining this linked-list order is in fact necessary to make our inductive approach work. Instead of a standard $\Insert(x)$ operation, we have two operations, $\PushLeft(x)$ and $\PushRight(x)$ that append the item to the appropriate end of the linked list. Furthermore, we also implement operations \First{} and \Last{} that return respectively the leftmost and rightmost item, as well as operations $\Next(x)$ and $\Prev(x)$ that allow us to traverse the linked list.

Together with the above deque operations, our data structure also supports the standard heap operations $\DecreaseKey{}(x, x')$, $\IncreaseKey{}(x, x')$, $\FindMin$ and $\Delete(x)$.\footnote{We choose to treat $\Delete$ as a fundamental heap operation instead of the more commonly used $\Deletemin$. In a heap supporting $\O(1)$-time $\DecreaseKey$ and $\FindMin$, both operations are equivalent up to additive constants.} The amortized cost of $\Delete(x)$ is the main feature of our construction: it is $\O(\log d_x)$, where $d_x$ is the \emph{deque size} of $x$ defined below. In particular, if we use our data structure as a deque and only ever delete the left-/rightmost item, the cost of \Delete{} is $\O(1)$.

\begin{definition}[Stack, queue and deque size]
	For a heap $H$ having the interface described above, define the \emph{stack size} ($s_x$) and \emph{queue size} ($q_x$) of $x$ as, respectively, the distance of $x$ to $H.\First$ and $H.\Last$, plus one. Define the \emph{deque size} of $x$ as $d_x = \min(s_x, q_x)$.
\end{definition}

Note that the standard way to define $s_x$, $q_x$ and $d_x$ is by counting the number of items newer/older than $x$. In our interface, we can recover those definitions by only allowing $\PushLeft$ and not $\PushRight$: then the linked-list order is exactly the insertion-time order. By allowing $\PushRight$, our heap and definitions are slightly stronger. We need this flexibility (and the ability to perform linked-list operations) in our inductive construction.

We are ready to formally state the guarantees of our construction:

\begin{theorem}
    \label{thm:heap_main}
    For any fixed positive integer $k$, there exists a data structure in the comparison model that supports the following operations in the following amortized time complexity. 
    \begin{enumerate}[itemsep=0pt]
		\item $\DecreaseKey(x, x')$, \FindMin, $\Next(x)$, $\Prev(x)$, \First, \Last{} in time $\O(1)$, 
		\item $\PushLeft(x)$, $\PushRight(x)$ in time $\O(\log^{*(k)}(n))$, 
		\item $\IncreaseKey(x, x')$ and $\Delete(x)$ in time $\O(\log d_x + \log^{*(k)}(n))$.
    \end{enumerate}
	Here $\log^{*(k)}$ is the $k$-fold iterated logarithm, that is, the $k$-th function in the sequence $\log^*$, $\log^{**}$, $\log^{***}$, \dots{}, formally defined in \cref{def:k-fold}.
\end{theorem}

Note that the amortized cost of $\Delete(x)$ can be reduced to $\O(\log d_x)$ by charging $\PushLeft$ and $\PushRight$ (see \cref{rem:delete_no_overhead}). With this, it is easy to see that \cref{thm:heap_main} implies the informal \cref{thm:heap_main_informal}.

\paragraph{Roadmap}
We prove \cref{thm:heap_main} in the rest of this section. First, in \cref{subsec:induction}, we state the main inductive lemma, \cref{lem:inductive}, that implies \cref{thm:heap_main}. Next, in \cref{subsec:quartermaster} we explain the main building block of our data structure that we call a \emph{quartermaster}. Finally, \cref{sec:heap_finish} finishes the proof by discussing how the quartermasters are used to build the heap from \cref{lem:inductive}.

\subsection{Main inductive lemma}
\label{subsec:induction}
The main lemma that we prove in this section is an inductive construction: We start with a heap that is close to having the deque-like property with fast $\Insert$ (paying an additive $\O(\log^{*(k)}n)$ factor) and that creates a new data structure that is even closer to having the property (paying an additive $\O(\log^{*(k+1)}n)$ factor). We start with defining a heap that is close to having the deque-like property.

\begin{definition}[$f$-deque-like, $f$-stack-like heap]
	\label{def:f-deque-like}
	Let $f$ be a function from the following sequence:
	\begin{equation}
		\label{eq:f-iterated-log-like}
		\frac 12 \log n, \log \log n, \log^{*} n, \log^{**} n, \ldots, \log^{*(k)} n, \ldots
	\end{equation}
	We say a heap is \emph{$f$-deque-like} if it supports the following operations in the following amortized time complexity, where $n$ is the number of items currently in the heap:
	\begin{enumerate}[itemsep=0pt]
		\item $\DecreaseKey(x, x')$, \FindMin, $\Next(x)$, $\Prev(x)$, \First, \Last{} in time $\O(1)$, 
		\item $\PushLeft(x)$, $\PushRight(x)$ in time $\O(f(n))$, 
		\item $\Delete(x)$ and $\IncreaseKey(x, x')$ in time $\O(\log d_x + f(n))$.
    \end{enumerate}

	An $f$-stack-like heap is defined analogously, with $d_x$ being replaced by $s_x$, and without requiring the support for $\PushRight(x)$.
\end{definition}

Now we can state the inductive lemma:

\begin{lemma}
    \label{lem:inductive}
	If there is an $f$-deque-like heap, there is also a $g$-deque-like heap, where $f, g$ is a pair of two consecutive functions in \cref{eq:f-iterated-log-like}.
\end{lemma}

We will now verify that there is a $(\frac 12\log n)$-deque-like heap. This will serve as the base case of our construction. This data structure is essentially just the Fibonacci heap with an additional linked list for maintaining the deque order. 

\begin{lemma}[Maintaining deque order with a heap]
	\label{lem:heaps_fibonacci_deque}
	There is a $(\frac 12\log n)$-deque-like heap.
\end{lemma}
\begin{proof}
	First, the $\frac12$ factor does not matter as $\O(f(n)) = \O(\log n)$.
	We store the items in a Fibonacci heap while additionally maintaining a linked-list order on the items. That is, a~record of item $x$ in the Fibonacci heap is a triple ($x$, \textit{pointer to the item left of~$x$}, \textit{pointer to the item right of $x$}). The heap also maintains pointers to its leftmost and rightmost item. Whenever we call $\PushLeft$, $\PushRight$ and $\Delete$, we perform the Fibonacci heap operation ($\Insert$ in case of $\PushLeft$ and $\PushRight$, $\Delete$ in case of $\Delete$) and then update the linked list pointers and the pointers to the leftmost/rightmost item, as we would in a proper linked list.\footnote{In the case of $\Delete(x)$, we need to back up the record of $x$ before doing the heap $\Delete$, since afterwards we would lose the access to the linked list pointers and would be unable to update the linked list properly.} $\IncreaseKey(x, x')$ is implemented as $\Delete$ followed by $\Insert$ on the heap (while keeping the item's position on the linked list). The implementation of other operations is straightforward, as those operate purely on the heap or purely on the linked list structure.
\end{proof}

Now we can verify that \cref{thm:heap_main} follows from \cref{lem:inductive}. 
\begin{proof}[Proof of \cref{thm:heap_main}]
    To achieve the data structure from the statement of \cref{thm:heap_main} parameterized by $k$, we use \cref{lem:heaps_fibonacci_deque} and apply \cref{lem:inductive} $k+1$ times. The complexity of $\Delete$ and $\IncreaseKey$ then goes from $\O(\log d_x + \log n)$ (from \cref{lem:heaps_fibonacci_deque}) to $\O(\log d_x + \log\log n)$ (first application of \cref{lem:inductive}) to $\O(\log d_x + \log^{*(k)}(n))$ ($k$ more applications of \cref{lem:inductive}). Similar analysis also holds for $\PushLeft$ and $\PushRight$.
\end{proof}

\begin{remark}[No free lunch]
	One might ask: why not just repeat \cref{lem:inductive} indefinitely and drive the overhead down to $\O(1)$? The answer is that each application of \cref{lem:inductive} comes at the cost of worse constant factors in the complexities of other operations; hence we apply the lemma only $\O(1)$ times in the proof. Thus picking $k = \alpha(n)$ would lead to a $2^{\O(\alpha(n))}$-deque-like heap (where $\alpha(n)$ is the inverse Ackermann function). It is an interesting question whether a better analysis can show that this gives an $\alpha(n)$-deque-like heap.
\end{remark}

\begin{remark}[Heap pointers]\label{rem:pointerless}
	In our construction, our items are scattered across multiple heaps, and we need to answer queries of the form ``what heap is the item $x$ in?''. This can be problematic in general -- in particular, when the heap has to support the $\Merge$ operation, then this is as hard as the disjoint-set-union problem. However, in our case, we do not need heap merging. Thus we simply augment each item to store the pointer to the heap it is in. This gets assigned during the item's $\Insert$, and updated whenever we internally move items across heaps. This also works recursively: if items are stored in heaps, which are stored in heaps of heaps, and so on, then each item gets a pointer to its heap, each heap gets a pointer to its heap of heaps, and so on.
\end{remark}

\begin{remark}[Properties of $f$]\label{rem:f-properties}
	In our construction of $f$-deque-like heaps, we use subtle assumptions on $f$ being ``reasonable''. For example, one can verify that any $f$ from \cref{eq:f-iterated-log-like} is non-decreasing, and we have $f(\O(n)) = \O(f(n))$, and thus also e.g.~$f(n + 1) = \O(f(n))$.
\end{remark}

In the rest of this section, our goal is to prove \cref{lem:inductive} by assuming the existence of an $f$-deque-like heap (as defined in \cref{def:f-deque-like}) and trying to construct a $g$-deque-like heap for $g = o(f)$ defined in \cref{lem:inductive}.

\subsection{Bundles, warehouses, and quartermasters}
\label{subsec:quartermaster}

In this subsection, we describe three helpful data structures that we call a \emph{bundle}, a \emph{warehouse} and a \emph{quartermaster}. Each data structure has roughly the same interface as the final heap we are going to build, but each exhibits a different set of tradeoffs in their time complexities and restrictions on how they can be used. Together, these tradeoffs combine in a way that makes the whole induction possible. On a high level, our final data structure will be a collection of quartermasters, arranged next to each other, each of which internally organizes its items in a warehouse, which is a collection of bundles. The whole construction has two main goals: preserve the $\O(\log d_x)$ beyond-worst-case term in the amortized cost of $\Delete(x)$ and $\IncreaseKey(x, x')$, while reducing the additive overhead from $\O(f(n))$ to $\O(g(n))$.

The outline of the construction is displayed in \cref{fig:heap2}, which is a more general version of \cref{fig:heap}.

\begin{figure}
    \centering
    \includegraphics[width=\linewidth]{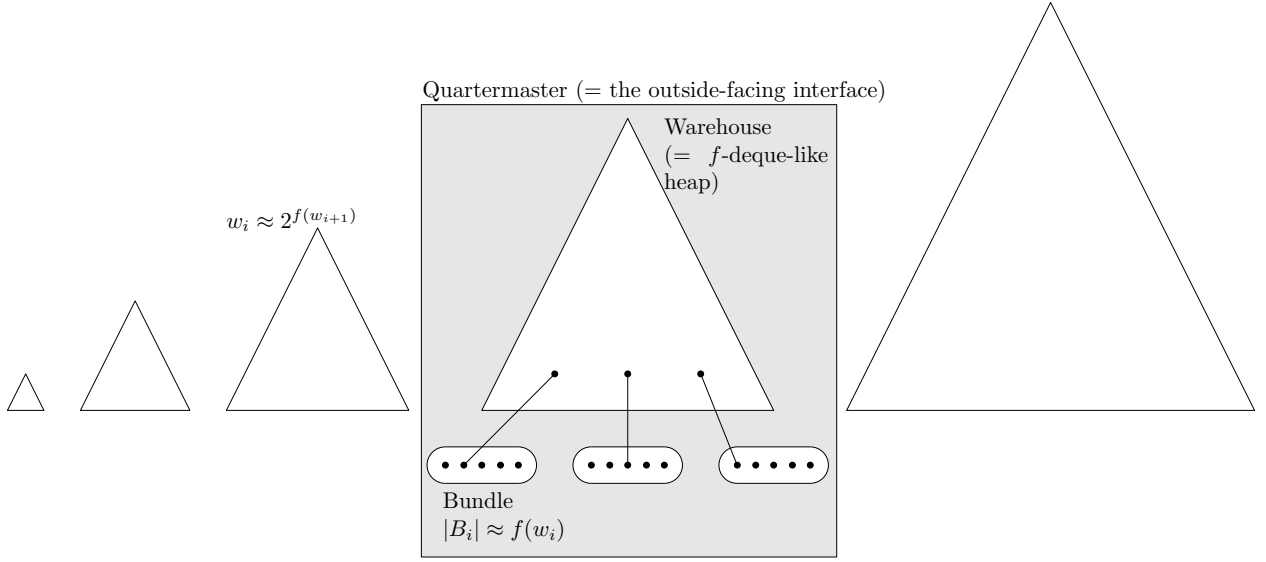}
	\caption{Diagram of our recursive heap construction.
The data structure consists of quartermasters of fast-growing sizes. The construction aims to preserve the linked-list order of the items, newly inserted items are thus inserted in the leftmost, smallest, quartermaster. \\
	Internally, each quartermaster groups its items into bundles of size roughly $f(w_i)$, where $w_i$ is the preferred size of the $i$-th quartermaster, and stores the bundles, keyed by their minima, in an already-constructed $f$-deque-like heap called a warehouse.}
    \label{fig:heap2}
\end{figure}

\paragraph{Bundle} We start by defining a very simple data structure that we call a bundle.

\begin{claim}[Bundle]
    \label{cl:bundle}
    There exists a data structure that stores a collection $B$ of items. It implements the following operations in the following time. 
    \begin{enumerate}
        \item $\DecreaseKey(x, x')$, $\FindMin$, $\Next(x)$, $\Prev(x)$, $\First$, $\Last$, $\PushLeft(x)$, \PushRight$(x)$ in $\O(1)$ worst-case time per operation.
		\item $\IncreaseKey(x, x')$ and $\Delete(x)$ in $\O(|B|)$ worst-case time per operation, where $|B|$ is the number of items stored in the bundle.
    \end{enumerate}
\end{claim}
\begin{proof}
    We can implement a bundle using a doubly-linked list that maintains the current minimum. $\Delete(x)$ and $\IncreaseKey(x, x')$ first delete/increase $x$ and then iterate over the whole bundle to recalculate the minimum.
\end{proof}

\paragraph{Quartermaster}
%
%
%
Next, we define a data structure that we call a quartermaster. It can be thought of as a specialized version of an~$f$-deque-like heap that internally stores items in bundles. The resulting heap is still an~$f$-deque-like heap (and not $o(f)$-deque-like), but it allows (limited) constant-time manipulation with items at both ends of the linked-list order (as opposed to $\O(f(n))$ guaranteed by $f$-deque-likeness).

This is crucial for our construction. The final $g$-deque-like heap will consist of several quartermasters, each taking care of a segment of the linked list order. The quartermasters will frequently pass items to each other, and thus our goal is to make insertion/deletion to/from both ends of the linked list run in amortized $\O(1)$ time per item. We pay for this in slight inflexibility: faster deletion only applies when we delete in \emph{batches}, and the whole data structure has a maximum size limit $w$ that must be set in advance.

In the rest of this section, we prove the following lemma, which specifies the interface of a quartermaster.

\begin{lemma}[Quartermaster]
    \label{lem:quartermaster}
	Assume there exists an $f$-deque-like heap for some $f$.
	Then for any positive integer parameter $w$, there exists a data structure that we call a \emph{quartermaster} and that implements the following operations in the following time. At any point in time, the quartermaster can store at most $\O(w)$ items.
	\begin{enumerate}[itemsep=0pt]
	\item \Next, \Prev, \First, \Last, \FindMin{}: in $\O(1)$ worst-case time per operation.
        \item \DecreaseKey{}, $\PushLeft$, $\PushRight$: in $\O(1)$ amortized time per operation. 
        \item $\Delete(x), \IncreaseKey(x, x')$: in $\O(\log d_x + f(w))$ amortized time per operation.
        \item \BatchPopLeft{} and \BatchPopRight{}: in $\O(|L|)$ amortized time per operation, where $L$ is the list of the items being returned by this operation.
    \end{enumerate}
The operation \BatchPopLeft{} extracts some unknown number of items that is between 1 and $f(w)$ from the data structure that are leftmost in the deque order and returns them, correctly ordered, on output. The operation \BatchPopRight{} works analogously.
\end{lemma}

Before proving the lemma, we start with a remark that simplifies our analysis:

\begin{remark}
	\label{rem:delete_no_overhead}
	Any $f$-deque-like heap can be assumed to perform $\Delete(x)$ (but not $\IncreaseKey(x, x')$) in amortized time $\O(\log d_x)$, without the additive $f(n)$ term, while preserving the time complexity of all other operations. This is because we can charge the additive $f(n)$ terms of all $\Delete$ operations to the cost of all $\PushLeft$ and $\PushRight$ operations. Concretely, we can define a potential $\Phi = f(1) + f(2) + \ldots + f(n)$ where $n$ is the heap's current size. $\Delete(x)$ frees up $f(n)$ potential and $\PushLeft/\PushRight$ has to pay for an additional $f(n + 1) = \O(f(n))$ potential increase.
\end{remark}

\paragraph{Roadmap}
In the rest of \cref{subsec:quartermaster}, we prove \cref{lem:quartermaster}. First we describe the data structure. Then, we describe how to implement various operations in a series of claims.

\paragraph{Data structure}
The items of a quartermaster are stored in bundles, where each bundle contains a sequence of items that are consecutive in the deque order. Moreover, we use the inductive $f$-deque-like heap to store the bundles, using their minimum as keys.\footnote{We might need to temporarily store (at most one) empty bundle, whose key we define as $+\infty$.} We call this heap the \emph{warehouse} of the quartermaster. See \cref{fig:heap2}. We maintain the invariant that the quartermaster contains at most $\O(w)$ items in total (for $w$ set in advance by the caller of \cref{lem:quartermaster}). Thus the warehouse also contains at most $\O(w)$ bundles in total, and the additive overhead of operations on the warehouse heap is $\O(f(\O(w))) = \O(f(w))$.

\paragraph{Invariants and potential function}
The quartermaster will keep the following invariant: We define $b = f(w)$ and keep that each bundle is of size at least 1 and at most $b$.
To analyze this data structure, we will also use the following potential $\Phi$. For each bundle $B$, we define its potential 
\[
\phi(B) = \max\left(0,\;|B| - \left\lceil\frac b2\right\rceil\right)
\]
and we define the overall potential of the quartermaster as the sum of bundle potentials over all bundles:
\[
\Phi = \sum_{\text{$B$}} \phi(B).
\]

\paragraph{Quartermaster operations}

Next, we show how quartermaster operations can be implemented; we also analyze their amortized complexities with respect to the potential $\Phi$.  (In the end, we scale $\Phi$ up by an appropriate constant so that it pays for the operations correctly.)

\begin{claim}
    \label{cl:simple_operations}
	The following operations can be implemented in $\O(1)$ amortized time: \Next, \Prev, \First, \Last, \DecreaseKey{}, \FindMin{}.
\end{claim}
\begin{proof}
    We implement the operations as follows:
	\begin{itemize}[itemsep=0pt]
		\item \Next{}/\Prev{}: get the next/previous item in the bundle, if it exists. Otherwise, get the next/previous bundle in the warehouse, and return its first/last item in this bundle, if both exist. Otherwise fail.
		\item \First{}/\Last{}: return the first/last item of the first/last bundle.
		\item $\DecreaseKey(x, x')$: call $\DecreaseKey(x, x')$ on the bundle $B$ containing $x$, then call $\DecreaseKey(B, B.\FindMin)$ on the warehouse.
		\item \FindMin: call $\FindMin{}$ on the warehouse, then call $\FindMin{}$ on the returned bundle.
	\end{itemize}

	None of the operations change the potential $\Phi$. Their amortized cost is thus their worst-case cost, which is $\O(1)$.
\end{proof}

\begin{claim}
    \label{cl:quartermaster_pushright}
    The operations $\PushLeft(x)$ and $\PushRight(x)$ can be implemented in $\O(1)$ amortized time.
\end{claim}
\begin{proof}
	We will describe and analyze $\PushRight(x)$; $\PushLeft(x)$ is analogous. The operation is described in \cref{alg:quartermaster_pushright}.
\begin{algorithm}
\caption{quartermaster.\PushRight(x)}
\label{alg:quartermaster_pushright}
\begin{algorithmic}[1]
\State $W$ is the $f$-deque-like warehouse storing the bundles 
\If{$W$ is empty}
\State insert an empty bundle into $W$\Comment{$\O(1)$}
\EndIf
\State $B \gets W.\Last$
\If{$|B| = b$}
	\State $W.\Delete(B)$ \Comment{$\O(\log d_B' + f(w)) = \O(f(w))$}
	\State $B_1, B_2 \gets $ split $B$ into two bundles, $B[:\lceil b/2\rceil]$ and $B[\lceil b/2\rceil:]$\Comment{$\O(b) = \O(f(w))$}
	\State $W.\PushRight(B_1)$, $W.\PushRight(B_2)$ \Comment{$\O(f(w))$}
	\State $B \gets B_2$
\EndIf
\State $B.\PushRight(x)$
\State $W.\DecreaseKey(B, B.\FindMin)$
\end{algorithmic}
\end{algorithm}

	In the special case where $W$ starts out empty, we insert an empty bundle in $\O(f(|W|)) = \O(f(0)) = \O(1)$ time.

The $B.\PushRight(x)$ call at the end runs in $\O(1)$ worst-case time and only changes the potential by a constant. The $\DecreaseKey$ call runs in $\O(1)$ time.

	What remains is to bound the amortized time complexity of the split step when $|B| = b$. $W.\Delete(B)$ takes $\O(\log d_B' + f(w)) = \O(f(w))$ time, where $d_B' = 1$ is the deque size of $B$ in the warehouse heap. The split step takes $\O(|B| + f(w)) = \O(b)$ worst-case time. Before the split, the potential of the bundle was $\max(0, |B| - \lceil b/2\rceil) = b - \lceil b/2 \rceil = \lfloor b/2 \rfloor$. The potential of the newly created bundles is $0$, as their size is not greater than $\lceil b/2\rceil$. Thus, the operation freed up $\Theta(b)$ potential and the total amortized cost is $\O(1)$. 
\end{proof}

\begin{claim}
    \label{cl:quartermaster_extract}
    The operations $\Delete(x)$ and $\IncreaseKey(x, x')$ can be implemented in amortized time $\O(\log d_x + f(w))$. 
\end{claim}
\begin{proof}
	The implementation of $\Delete(x)$ is described in \cref{alg:quartermaster_extractmin}. On Line~2, we use our ability to query an item's bundle, guaranteed by \cref{rem:pointerless}.

\begin{algorithm}
\caption{quartermaster.\Delete$(x)$}
\label{alg:quartermaster_extractmin}
\begin{algorithmic}[1]
\State $W$ is the $f$-deque-like warehouse storing the bundles
\State $B \gets \text{the bundle containing $x$}$
\State $B.\Delete(x)$ \Comment{$\O(|B|) = \O(f(w))$}
\If{$|B| = 0$}
	\State $W.\Delete(B)$ \Comment{$\O(\log d_x + f(w))$}
\Else
	\State $W.\IncreaseKey(B, B.\FindMin)$ \Comment{$\O(\log d_x + f(w))$}
\EndIf
\State\Return $x$
\end{algorithmic}
\end{algorithm}
	The bundle potential only changes by a constant. It is thus enough to bound the cost of the $W.\Delete$ and $W.\IncreaseKey$ call. This is $\O(\log d_B' + f(w))$, where $d_B'$ is the deque size of the bundle $x$ is contained in (measured in the warehouse heap). We have $d_B' \le d_x$, and thus the cost is $\O(\log d_x + f(w))$, as needed.

 The $\IncreaseKey(x, x')$ operation is implemented similarly, except that the $|B| = 0$ case never happens. Thus the amortized cost is also $\O(\log d_x + f(w))$.
\end{proof}

\begin{claim}
    \label{cl:quartermaster_batch}
    The operations \BatchPopLeft{}, \BatchPopRight{} can be implemented in $\O(|B|)$ amortized time where $B$ is the bundle being returned.
\end{claim}
\begin{proof}
	\BatchPopLeft{}/\BatchPopRight{} simply delete and return the leftmost/rightmost bundle from the warehouse. By \cref{rem:delete_no_overhead}, we can assume that the warehouse supports $\Delete$ in $\O(\log d_B)$. This is $\O(1)$, as $B$ is the leftmost/rightmost bundle. Additionally, the overall potential only decreases.
\end{proof}

Let us finish the proof of \cref{lem:quartermaster}:

\begin{proof}[Proof of \cref{lem:quartermaster}]
	First, we notice that all operations keep the necessary invariant that all bundles contain between 1 and $b$ items. Next, the claimed amortized complexity of operations follows from \cref{cl:simple_operations,cl:quartermaster_pushright,cl:quartermaster_extract,cl:quartermaster_batch}. This finishes the proof. 
\end{proof}

\subsection{Building from quartermasters}
\label{sec:heap_finish}

With bundles and quartermasters as our building blocks, we now turn our attention to finishing the construction. We start by proving that we only need to achieve the $g$-stack-like property, which simplifies the analysis.

\subsubsection{From stack-like to deque-like}

Here we present a simple construction that turns any $f$-stack-like heap into an $f$-deque-like heap. Its caveat is that even if the old heap had worst-case guarantees, the new heap only has amortized guarantees. Recall that an $f$-stack-like heap supports $\PushLeft(x)$, but not $\PushRight(x)$, and that the cost of $\Delete(x)$ and $\IncreaseKey(x, x')$ is $\O(\log s_x + f(n))$, where $s_x$ is the number of items left of $x$ (including $x$ itself). 

Although we were not able to find this exact black-box construction in the literature, the general idea is standard.

\begin{lemma}
	\label{lem:stack_implies_deque}
	If there is an $f$-stack-like heap, then there is an $f$-deque-like heap.
\end{lemma}
\begin{proof}
	The construction works as follows: We maintain the items in two $f$-stack-like heaps, $H_L$ and $H_R$, such that the linked-list order is ``items from $H_L$ in left-to-right order, followed by items from $H_R$ right-to-left order''. At all times, we maintain that the sizes of both heaps satisfy $|H_L|, |H_R| \ge \lfloor n/4\rfloor$. Whenever this condition is violated, we simply dismantle both structures, capture the linked list order and build $H_L$ and $H_R$ anew with $\lfloor n/2\rfloor$ and $\lceil n/2\rceil$ items respectively.

Whenever we are rebalancing a heap of size $n$, at least $\Omega(n)$ $\PushLeft$, $\PushRight$ and $\Delete$ operations must have happened since the last rebuild. The cost of rebalancing is $\O(n \cdot f(n))$ and we can pay for it by charging each past operation since the last rebuild an additional $\O(f(n))$ cost. Every operation is charged at most once in this way.

	We can straightforwardly implement all operations: $\PushLeft$, $\PushRight$, $\Delete$, \DecreaseKey{}, $\IncreaseKey$ pick the correct heap and perform the operation there. $\FindMin$ returns the smaller of the two heap minima, $\Next(x)$ and $\Prev(x)$ call the respective operation on the heap $x$ is in, possibly jumping into the other heap if necessary. $\First$ returns $H_L.\First$ and $\Last$ returns $H_R.\First$.

	Finally, we need to verify that the resulting heap is $f$-deque-like. Let $x$ be an item, without loss of generality in $H_L$, and let $s_x$ be its stack size in $H_L$. It is also the stack size in the compound heap. Let $d_x, q_x$ be the deque and queue size of $x$ in the compound heap, respectively. The $f$-deque-like property follows if we prove that $\log s_x = \O(\log d_x)$. We have $q_x \ge |H_R| \ge \lfloor n/4 \rfloor \ge \lfloor s_x / 4\rfloor$, where we used the fact that all of $H_R$ is right of $x$, the balancing property, and the fact that the stack size is smaller than the heap size. By definition, $d_x = \min(s_x, q_x) \ge \min(s_x, \lfloor s_x/4\rfloor) = \lfloor s_x/4\rfloor$. On the other hand, $d_x \le s_x$. Thus, $d_x$ and $s_x$ are only a constant factor apart, and indeed $\log s_x = \O(\log d_x)$ as needed.
\end{proof}

In the remainder of \cref{sec:heap_finish}, we will focus on implementing a $g$-stack-like heap.

\subsubsection{Data structure description}

Our data structure follows the sketch from \cref{fig:heap2}.

First, we start by defining a sequence of \emph{preferred sizes} $w_0, w_1, w_2, \dots$. This is done by setting $w_0$ to be a sufficiently large constant greater than $2$ and then inductively defining $w_{i+1}$ to be the smallest integer such that $2^{f(w_{i+1})} \ge w_i$. (From this it follows that $2^{f(w_{i+1} - 1)} < w_i$.)
Define $g(n)$ as the smallest integer $x$ such that $w_x \ge n$. The following lemma bounds $g$.

\begin{claim}
\label{cl:fstar}
	Assume a nondecreasing function $f$. If $f(n) \le 1/2 \cdot \log_2 n$ for $n \ge w_1$, then $w_{\log\log n} \ge n$. If $f(n) \le \log\log n$ for $n \ge w_1$, then $w_{2f^*(n)} \ge n$. In both cases, the sequence $w_i$ is increasing.

	Note that we can pick $w_0$ (and thus $w_1$) large enough so that the condition above always applies.
\end{claim}
\begin{proof}
    If $f(n) \le 1/2 \cdot \log_2 n$, we have
    \[
	w_i
	\le
	2^{f(w_{i + 1})}
	\le
	2^{1/2 \cdot \log_2 w_{i+1}}
	=
	2^{\log_2 \sqrt{w_{i + 1}}}
	=
	\sqrt{w_{i + 1}}
	,
    \]
	which, along with $w_0 \ge 2$ implies that $w_{\log_2\log_2 n} > n$.
	It also immediately follows that $w_{i + 1} \ge w_i ^ 2 > w_i$, which follows by induction on $w_0 \ge 2 > 1$.

    Next, if $f(n) \le \log_2\log_2 n$, we have that
	\[
	w_i
	\le
	2^{f(w_{i + 1})}
	\le
	2^{\log_2\log_2 w_{i+1}}
	=
	\log_2 w_{i + 1}
	\le
	\log_2 2^{f(w_{i + 2})}
	=
	f(w_{i + 2})
	\]
	Along with $w_0 \ge 2$, this implies that $w_{2f^*(n)} \ge n$. Additionally, we have $w_i \le \log_2 w_{i + 1} < w_{i + 1}$, as needed.
\end{proof}

In our heap, we maintain roughly $g(n)$ quartermasters $Q_1, Q_2, \ldots, Q_t$. (Note that $g(n)$ and $t$ change over time.) The linked list order is maintained such that the leftmost items are in $Q_1$ and the rightmost ones in $Q_t$. We also maintain the following size invariant:

\paragraph{Size invariant} The size of each quartermaster $Q_i$ satisfies $w_i/2 \le |Q_i| \le w_i$, except for the rightmost quartermaster $Q_t$, where only the upper bound $|Q_t| \le w_t$ needs to hold.

From \cref{cl:fstar}, it immediately follows that there are $\O(g(n))$ quartermasters. Also note that whenever $x \in Q_i$, then there are at least $w_{i-1} / 2$ items to the left of $x$. This is crucial for the stack-like property, as we will see later.

The size invariant can be temporarily violated in the middle of an operation. It will however still hold that $|Q_i| = \O(w_i)$, which is important for the time complexity term in $Q_i$ to still be $\O(f(w_i))$. We verify this later in \cref{cl:fix-bounded}.

In the remainder of \cref{sec:heap_finish}, we first describe the operations on our heap, which may temporarily break our size invariant, and then we describe how we can maintain the size invariant, and analyze the maintenance cost.

\subsubsection{Heap operations}

We are ready to describe the heap operations. First, we describe all operations that do not change how items are distributed across quartermasters:

\begin{claim}[Read-only operations]
\label{cl:main-read-only}
	We can implement the functions $\FindMin$, $\Next(x)$, $\Prev(x)$, $\First$, $\Last$ in $\O(1)$ worst-case time.
\end{claim}
\begin{proof}
	$\FindMin$: We maintain the current global minimum and modify it during all operations: for $\PushLeft$ and $\DecreaseKey$, we check in $\O(1)$ time whether the minimum has decreased; for $\Delete$ and $\IncreaseKey$, we iterate over all quartermasters and recompute the new minimum in $\O(g(n))$ time.

	$\Next(x)$ and $\Prev(x)$: we call the analogous operation on the quartermaster $x$ is in, possibly jumping to the next/previous quartermaster and returning the first/last item there, if necessary.

	$\First$ and $\Last$: we call $Q_1.\First$ and $Q_t.\Last$, respectively.
\end{proof}

\begin{claim}
\label{cl:main-decrease}    
We can implement the operation $\DecreaseKey(x, x')$ in amortized complexity $\O(1)$. 
\end{claim}
\begin{proof}
	We call $\DecreaseKey(x, x')$ in the appropriate quartermaster storing $x$ (that we find using \cref{rem:pointerless}). Then we update the global minimum in $\O(1)$ time.
\end{proof}

Now we describe $\PushLeft(x)$ and $\Delete(x)$, which change how items are distributed across quartermasters. Namely, we insert/delete the item to/from an appropriate quartermaster, and as a result, its size constraint may become tight. We fix this by running one of two rebalancing procedures, which we call $\FixTooBig$ (used after $\PushLeft$), and $\FixTooSmall$ (used after $\Delete$). We define and analyze both later.

\begin{claim}
\label{cl:main-pushleft}
We can implement the operation $\PushLeft(x)$ such that its amortized complexity is $\O(1)$, plus the work done by $\FixTooBig$.
\end{claim}
\begin{proof}
	We run $Q_1.\PushLeft(x)$. This takes $\O(1)$ amortized time. Then we update the global minimum in $\O(1)$ time. Finally, we run $\FixTooBig$.
\end{proof}

\begin{claim}
\label{cl:main-delete}    
	We can implement the operations $\Delete(x), \IncreaseKey(x, x')$ such that their amortized complexity is $\O(\log s_x + g(n))$, plus the work done by $\FixTooSmall$ (in the case of $\Delete$).
\end{claim}
\begin{proof}
	We apply the corresponding operation on the quartermaster $Q_i$ that $x$ is in (found using \cref{rem:pointerless}). Then we recalculate the global minimum in $\O(g(n))$ time by iterating over all quartermasters. Finally, we run $\FixTooSmall$.

	The cost of the quartermaster operation is $\O(\log s_x' + f(w_i))$ where $s_x' \le s_x$ is the stack size of $x$ in $Q_i$. Note that $s_x \ge |Q_{i-1}| \ge w_{i-1} / 2$, and by using the definition of $w_i$ as the smallest integer such that $2^{f(w_i)} \ge w_{i-1}$, we get $2^{f(w_i - 1)} < w_{i-1}$, and thus $s_x > 2^{f(w_i - 1)} / 2$ and $\log s_x > f(w_i - 1) - 1$.
	The complexity is thus $\O(\log s_x + f(w_i)) = \O(\log s_x + f(w_i - 1)) = \O(\log s_x)$. In the first part, we use \cref{rem:f-properties} to conclude that $f(w_i) = \O(f(w_i - 1))$.
\end{proof}

\subsubsection{Maintaining the size invariant}
\label{sec:size_invariant}

What remains to be shown is how to maintain the size invariant using the $\FixTooBig$ and $\FixTooSmall$ operations.

Both operations are implemented in the most straightforward way possible. $\FixTooBig$ iterates with $i = 1, \ldots, t$, and for a fixed $i$, it repeatedly deletes the rightmost bundle from $Q_i$ (using $\BatchPopRight$) and gives it to $Q_{i + 1}$ (using repeated $\PushLeft$), until $|Q_i| \le w_i$. $\FixTooSmall$ similarly iterates with $i = 1, \ldots, t - 1$, and for a fixed $i$, deletes a bundle from $Q_{i + 1}$ (using $\BatchPopLeft$) and gives it to $Q_i$, until $|Q_i| \ge w_i/2$. Whenever we want to give items to a nonexistent quartermaster $Q_{t+1}$, we create it and increment $t$, and whenever the last quartermaster gets empty, we delete it and decrement $t$. We also note that our bounds may break for the first $\O(1)$ quartermasters, which is not an issue, since they only store $\O(1)$ items anyway, and we can handle those items separately if needed.

Since the four quartermaster operations $\PushLeft$, $\PushRight$, $\BatchPopLeft$ and \BatchPopRight{} all run in $\O(1)$ time per item inserted/deleted, we can bound the total time complexity of $\FixTooBig$ and $\FixTooSmall$ by the total number of item moves across the whole lifetime of the heap.

\paragraph{Proof strategy} 
Fix a quartermaster $Q_i$ and denote by $F_i$ the collection of items in the leftmost $i$ quartermasters at a given time. An item can enter $F_i$ in two ways: by being inserted into $Q_1$ during $\PushLeft$, or by being moved from $Q_{i+1}$ to $Q_i$ during $\FixTooSmall$. Similarly, it can leave $F_i$ either by being deleted from a quartermaster during $\Delete$, or by being moved from $Q_i$ to $Q_{i+1}$ during $\FixTooBig$. The high-level idea is that the number of $\PushLeft$ and $\Delete$ calls is known, and we can use them to bound the number of item moves.



We start by proving that if a $\FixTooBig$ or $\FixTooSmall$ modified a quartermaster $Q_i$, then even after the operation, the size of $F_i$ must be right at the top/bottom size threshold. First a technical lemma about $|F_i|$ being dominated by $|Q_i|$:

\begin{claim}
	\label{cl:w_i-dominates}
	It holds that $w_1 + \ldots + w_{i - 1} + f(w_{i+1}) = o(w_i)$.
	At any given point, it holds that $|Q_1| + |Q_2| + \ldots + |Q_{i-1}| = o(w_i)$.
\end{claim}
\begin{proof}
	We can bound the first $i - 1$ terms: By \cref{cl:fstar}, we have either $w_j \le \sqrt{w_{j + 1}}$ or $w_j \le \log{w_{j+1}}$. In either case, we get $\sum_{j=1}^{i-1} w_j = \O(w_{i-1}) = \O(\sqrt{w_i}) = o(w_i)$.

	For the last term, use the definition: $w_i \ge 2^{f(w_{i+1} - 1)}$. Taking a logarithm on both sides: $f(w_{i + 1} - 1) \le \log w_i$, so $f(w_{i+1}) = \O(\log w_i) = o(w_i)$ by \cref{rem:f-properties}. 

The second half of the claim follows immediately from the fact that $|Q_j| \le w_j$.
\end{proof}

\begin{claim}
\label{cl:fix-bounded}
	At a given time, let $F_i = \bigcup_{j=1}^i Q_j$. After a $\FixTooBig$ call that decreased $|F_i|$, we have $w_i - o(w_i) < |F_i|$. After a $\FixTooSmall$ call that increased $|F_i|$, we have $|F_i| < w_i/2 + o(w_i)$.

	At all times (even in the middle of an operation), we have $|Q_i| = \O(w_i)$.
\end{claim}
\begin{proof}
	Let us have a $\FixTooBig$ call that decreases $|F_i|$. This means some items have moved from $Q_i$ to $Q_{i+1}$. At the beginning, we necessarily have $|Q_i| > w_i$. We then repeatedly move $x \in [1, f(w_i)]$ items from $Q_i$ to $Q_{i+1}$ until $|Q_i| \le w_i$. It follows that when we stop, we have $w_i - f(w_i) < |Q_i| \le |F_i|$. We immediately get $|F_i| > w_i - o(w_i)$.

	We can proceed similarly in $\FixTooSmall$: if it increases $|F_i|$, then we start with $|Q_i| < w_i/2$ and then we repeatedly increment it by $x \in [1, f(w_{i+1})]$. Thus we end up with $|Q_i| < w_i/2 + f(w_{i + 1})$, and adding the size bound, we obtain: $|F_i| < (w_1 + \ldots + w_{i-1} + f(w_{i+1})) + w_i/2 = w_i/2 + o(w_i)$.

	Finally, the only time we have $|Q_i| \ge w_i$ is temporarily during $\FixTooBig$. We can show by induction on $i$ that even then, we have $|Q_i| \le 1 + w_0 + w_1 + \ldots + w_i = w_i + o(w_i) = \O(w_i)$.
\end{proof}

\begin{lemma}
\label{lem:fix-cost}
	Let $i$ be a sufficiently large constant (dependent only on $f$). Consider the heap's lifetime until some time $x$.
	Let $M$ be the total number of item moves from $Q_i$ to $Q_{i+1}$ during this lifetime. Let $P$ be the number of $\PushLeft$ calls that happened while the heap had at least $i$ quartermasters. Then $P \ge M$.
\end{lemma}
\begin{proof}
	Let us look at the heap's history and cut it into maximal segments such that in each segment, $Q_i$ was always nonempty and there were no item moves from $Q_{i+1}$ to $Q_i$. Furthermore, trim each segment such that it ends on a $\FixTooBig$ that moves items from $Q_i$ to $Q_{i+1}$ right at the end of the segment. Segments with no $\FixTooBig$ are discarded. Let $P_s$ and $M_s$ be the number of pushes and item moves from $Q_i$ to $Q_{i+1}$ during this segment, respectively. We will prove that $P_s \ge M_s$, unless $i$ is a sufficiently small constant.

	Note that at the start of the segment, we have $|F_i| < w_i/2 + o(w_i)$. That is because either $Q_i$ was empty before the start of the segment and we have $|F_i| = o(w_i)$, or there was a $\FixTooSmall$ call that moved items from $Q_{i+1}$ to $Q_i$, after which $|F_i| < w_i/2 + o(w_i)$. On the other hand, at the end of the segment, a~$\FixTooBig$ call moved items from $Q_i$ to $Q_{i+1}$ and after it, we have $|F_i| \ge w_i - o(w_i)$.
	
	Every item move from $Q_i$ to $Q_{i+1}$ causes $|F_i|$ to decrease by one. $\Delete$s also decrease $|F_i|$. The only way $|F_i|$ can increase is by $\PushLeft$ -- crucially, that is because we defined segments so that there are no moves from $Q_{i+1}$ to $Q_i$. Thus, at least $(w_i - o(w_i)) - (w_i/2 + o(w_i)) + M_s = w_i/2 + M_s - o(w_i)$ $\PushLeft$ operations must have happened in this segment, and we have that $P_s - M_s \ge w_i/2 - o(w_i)$. Note that the right-hand side depends only on $i$ and $f$. Since we assume that $i$ is a sufficiently large constant for a given $f$, we have $P_s - M_s \ge 0$. By summing this up over all segments, we get the needed bound $P \ge M$. Note that the segments were disjoint, we summed over all moves and over a subset of $\PushLeft$s, and thus the bound is valid.
\end{proof}

\begin{lemma}
	\label{cl:total-fix-cost}
	The amortized cost of $\FixTooBig$ and $\FixTooSmall$ is $\O(t)$, where $t = \O(g(n))$ is the current number of quartermasters.
\end{lemma}
\begin{proof}
	First note that in order for an item to be moved from $Q_{i+1}$ to $Q_i$, it must, at some previous point, have been moved from $Q_i$ to $Q_{i+1}$. Thus, we only need to bound the $\FixTooBig$ cost and the $\FixTooSmall$ cost can be charged to it.

	\cref{lem:fix-cost} provides us with a charging scheme. For each $i$, either it is not sufficiently large for \cref{lem:fix-cost} to hold, in which case $w_i = \O(1)$ and the work on $Q_i$ per $\FixTooBig$ is also $\O(1)$ and we can directly charge it to the current $\PushLeft$. Or $i$ is sufficiently large. Then \cref{lem:fix-cost} guarantees that at any given time $x$, there were at least as many $\PushLeft$s to the heap while $Q_i$ was nonempty as there were item moves from $Q_i$ to $Q_{i+1}$. Thus, for any item move, we can charge $\O(1)$ an arbitrary such $\PushLeft$ that has not been charged yet. This pays for the item move and for the $\FixTooBig$. At the same time, each $\PushLeft$ is charged at most $t$ times where $t$ is the number of quartermasters at the time this $\PushLeft$ happens.
\end{proof}

Now we can finally prove \cref{lem:inductive}:

\begin{proof}[Proof of \cref{lem:inductive}]
	Having the $g$-stack-like property follows from \cref{cl:fstar,cl:main-read-only,cl:main-decrease,cl:main-pushleft,cl:main-delete,cl:total-fix-cost}. Having the $g$-deque-like property follows from \cref{lem:stack_implies_deque}.
\end{proof}

\section{Hierarchy of beyond-worst-case properties of heaps}
\label{s:weakstrong}

This section investigates the relationships between various working-set-like properties of heaps, as shown in \cref{fig:diagram}.
In \cref{sec:working_set_defs}, we define various working-set properties and show that they are all asymptotically equivalent. In \cref{sec:stack_like}, we discuss the stack-like property and show that it is strictly stronger than the working-set property. In \cref{sec:size_property}, we for completeness define the size property and show that it is strictly weaker than the working-set property.

We start with some useful definitions.

\paragraph{Framing}
To accommodate a wide range of heaps, including those that do not support \DecreaseKey{} or whose $\Insert$ cost is $\omega(1)$, our beyond-worst-case definitions only talk about $\ExtractMin$. This poses a definitional issue. Since all our definitions are amortized and we do not constrain this amortization, it is technically perfectly correct to call any kind of heap stack-like, since all heaps have an amortized $\ExtractMin$ cost of $0$ under the right amortization (that charges enough cost to $\Insert$s).

Thus, we instead frame all our equivalence theorems as ``any heap with $\ExtractMin$ cost $X$ can be reanalysed to have amortized $\ExtractMin$ cost $Y$ while incurring only $\O(1)$ additional amortized cost to all other operations.''

\paragraph{Sequence of requests}

Imagine we performed some operations on a heap and now we want to analyze their time complexity. The following definition strips away details unimportant for the analysis -- for example, whether an item was removed using $\Delete$ or $\ExtractMin$, whether the heap supports $\FindMin$, or whether an item was modified using $\DecreaseKey$ or $\IncreaseKey$.

\begin{definition}
\label{def:requests}
\noindent A \emph{sequence of requests} is a finite sequence, with three types of events:
\begin{compactenum}[a)]
\item $\Insert$ a new item into $H$,
\item $\Touch$ an item in $H$,
\item $\Delete$ an item from $H$.
\end{compactenum}
Initially, we start with an empty set $H$ and then we perform these requests. Touch requests do not modify $H$ and are only important for the analysis of some beyond-worst-case properties. Inserts and deletes also count as touches.

Each request has its associated time $t \in \{1, \ldots, m\}$. For an item $x$, we use $t_x$ to denote the time it was inserted and $t'_x > t_x$ the time it was deleted (possibly never, then $t'_x = \infty$).

	We use $H(t)$ to denote the state of $H$ at time $t$. That is, $H(t)$ is the set of items inserted, but not yet deleted, at $t$ (right \emph{before} the $t$-th operation, i.e., $x \notin H(t_x)$).
\end{definition}

Throughout this section we will use the heap formalism interchangeably with the formalism of \cref{def:requests}.

\subsection{Working-set properties}
\label{sec:working_set_defs}

Here, we discuss different possible definitions of the working-set property for heaps, and show that all of them are equivalent, up to constants and in an amortized sense. Thus, one can refer to all these definitions as ``the working-set property''. 

First, we list the working-set definitions ordered by strength and after each definition show that it is at least as powerful as the previous one. Then, in \cref{s:ws-equivalence}, we show that the weakest and strongest definition are equivalent, in the amortized sense and up to constant factors. Together, this implies the equivalence of all definitions.

\paragraph{Working-set property}
First, we define perhaps the weakest meaningful property that we call simply the \emph{working-set property}.

\begin{definition}[Working-set property]
\label{def:working_set}
    We say a heap has the \emph{(regular) working-set property}, if for every sequence of requests, the amortized complexity of $\ExtractMin$ is $\O(\log (t'_x - t_x))$.
\end{definition}

We note that this definition is perhaps the most similar to the definition of the working set in binary search trees. There one common definition of the working set size of $x$ is the number of operations since the last access of $x$. In our heap case, we instead use the weaker definition with time since \emph{insertion}. Otherwise, we would get a too strong property, at least with heaps that support fast $\DecreaseKey$:

\begin{remark}[Time since last touch is too strong \cite{elmasry_farzan_iacono}]
	Consider the following strengthened variant of the working set: In \cref{def:working_set}, instead of $t_x$, use the last time that $x$ was touched (i.e., inserted, decreased, increased etc.). For heaps that also support the \FindMin{} and $\DecreaseKey$ operation in $\O(1)$ time, we could then extract all items in constant time by first using \FindMin{} to identify them, then calling \DecreaseKey{} on them in $\O(1)$ time, and only then extracting them. 
\end{remark}

\paragraph{Touch-based working-set properties}
Next, we define three slightly stronger variants of the working-set property.
\begin{definition}[Insert-only / delete-only / touched-items working-set property]
\label{def:insert_delete_working_set}
    We say a heap has the \emph{insert-only working-set property}, if for every sequence of requests, the amortized complexity of $\ExtractMin$ is $\O(\log i_x)$ where $i_x$ is the number of insert operations that happened between $t_x$ and $t'_x$, inclusive. 

    The \emph{delete-only working-set property} is defined analogously, with $i_x$ replaced by $\delta_x$ that counts the number of delete operations.  

	Finally, the \emph{touched-items working-set property} is defined analogously, with $i_x$ replaced by $\tau_x$, the number of distinct items that were touched (i.e. inserted, deleted, or otherwise modified) by any operation between times $t_x$ and $t'_x$. 
\end{definition}

\begin{remark}
	\label{rem:touch_better_than_regular}
	\label{rem:insert_better_than_touch}
	All three properties of \cref{def:insert_delete_working_set} count a subset of touches between $t_x$ and $t'_x$, and are therefore at least as powerful as the regular working-set property. Furthermore, insert-only/delete-only working-set property are both at least as powerful as the touched-items working-set property, since $\tau_x$ is at least the number of inserted items and at least the number of deleted items.
\end{remark}

\paragraph{Strong working-set property}
The next property comes from the work of \citet{iacono} that initiated the study of working-set properties of heaps.

\begin{definition}[Stack of $x$, strong working set, strong working-set property]
\label{def:stack}
\label{def:strong_working_set}
	We define a \emph{stack of an item $x$ at time $t > t_x$}, or, $S_{x,t}$, to be the set of items that were inserted in the time interval $[t_x, t)$ and that are present in the heap at time $t$. Formally, $S_{x, t} = H(t) \setminus H(t_x)$. Note that $x$ is always in its stack. We use $S_x$ to denote the stack of $x$ at the time of its extraction, i.e., $S_x = S_{x, t'_x}$. We use $s_{x,t} = |S_{x,t}|$ and $s_x = |S_x|$.  

    We define the \emph{strong working set of an item $x$}, or $W_x$, to be its largest stack, that is $S_{x,t^*}$ for $t^* = \argmax_{t_x < t \le t_x'} |S_{x,t}|$.\footnote{If there are multiple such $t^*$, we pick an arbitrary (but fixed) one.} We use $w_x = |W_x|$.

We say a heap has \emph{the strong working-set property}, if for every sequence of requests, the amortized complexity of $\ExtractMin$ is $\O(\log w_x)$.
\end{definition}

\begin{remark}
	\label{rem:strong_better_than_insert}
	The strong working-set property is at least as strong as the insert-only property. This is because $W_x$ consists of a subset of items inserted between $t_x$ and $t'_x$, while the insert-only working-set property counts all such items.
\end{remark}

\paragraph{Strict working-set property}

Next, we define the even stronger variant of the strong working-set property. This strengthening is only relevant if the heap ends up non-empty.

\begin{definition}[Strict stack, strict working set, strict working-set property]
    \label{def:strict_working_set}
    We define the \emph{strict stack} $\Sigma_{x,t}$ and $\Sigma_x$ similarly to the (normal) stack $S_{x,t}$ and $S_x$, except that items that are never extracted from the heap are excluded from the stack. We analogously define $\sigma_{x,t} = |\Sigma_{x,t}|$ and $\sigma_x = |\Sigma_x|$.  

We define the \emph{strict working set of $x$}, or $\Strictw_x$, as its largest strict stack $\Sigma_{x,t^*}$ for $t^* = \argmax_{t_x < t \le t_x'} |\Sigma_{x,t}|$. We use $\strictw_x = |\Strictw_x|$.

We say a heap has \emph{the strict working-set property}, if for every sequence of requests, the amortized complexity of $\ExtractMin$ is $\O(\log \strictw_x)$.
\end{definition}

\begin{remark}
	\label{rem:strict_better_than_strong}
	The strict working-set property is at least as strong as the strong working-set property, since $\Sigma_{x,t} \subseteq S_{x,t}$ and thus $\strictw_x \le w_x$.
\end{remark}

\vasek{it would be helpful to add a picture with an example input that defines weak/normal/strict/stack-like on one example}

\subsubsection{Asymptotic equivalence of all working-set properties}
\label{s:ws-equivalence}

\cref{rem:touch_better_than_regular,rem:insert_better_than_touch,rem:strong_better_than_insert,rem:strict_better_than_strong} show that there is a hierarchy between the introduced working-set properties: regular is no stronger than touched-items, which is no stronger than insert-only, which is no stronger than strong, which is no stronger than strict.

In this section, we show that the regular working-set property is in fact asymptotically no weaker than the strict working-set property (\cref{thm:weak_strong_equivalence}) and the delete-only working-set property (\cref{thm:delete_implies_working_set}). This proves that all the introduced working-set properties are asymptotically equivalent, which finishes the proof of \cref{thm:all_same}.

\begin{theorem}
\label{thm:weak_strong_equivalence}
	Any heap with the working-set property can be shown to also have the strict working-set property, in the amortized sense, with only an additive $\O(1)$ amortized overhead to all operations.
\end{theorem}

\begin{proof}
	Fix any sequence $\rho$ of requests \Insert, \Touch, \Delete, as defined in \cref{def:requests}. Also, fix a heap with the working-set property. That is, the respective amortized cost of each $\Delete$ is $\O(\log (t'_x - t_x))$. Our goal is to set up an amortization argument under which each $\Delete$ cost is $\O(\log \strictw_x)$ and the cost of each operation only increases by $\O(1)$. 

    We classify each inserted item $x$ as either \emph{happy} or \emph{unhappy}; it is happy if $t'_x - t_x \le \strictw_x^3$. We note that for happy items, we have $\log(t'_x - t_x) \le 3 \log \strictw_x$. Hence, we only need to show how we can pay for the unhappy items. 

    This is done as follows. At the beginning, each operation in $\rho$ receives a constant potential $\O(1)$. We will now show how to redistribute this potential so that at the end, for every unhappy $x$, its $\Delete$ receives $\Omega(\log(t'_x - t_x))$ (actually even $(t'_x - t_x)^{\Omega(1)}$) potential from some past operations.\footnote{Every $\Delete$ only receives potential from the past, which is important, otherwise our potential could get negative and we would technically fall short of proving a true amortized cost guarantee.}

    More concretely, for every operation of $\rho$ that happened at time $t$, we consider the items whose $\Delete$ cost is affected by this operation: all items in the heap at time $t$ that will eventually be deleted. Denote them $x_1, x_2, \dots$, ordered in decreasing order by their insertion time. The operation sends a potential of $1/i^2$ to each item $x_i$. \vasek{add picture}

    On one hand, each operation overall sends away at most $1/1^2 + 1/2^2 + \dots = \O(1)$ potential. 
    On the other hand, consider any unhappy item $x$ inserted at time $t_x$ and deleted at time $t'_x$. We want to show that every operation at time $t_x < t \le t'_x$ sent enough potential to $x$. Fix one such $t$.
	It can be seen that $x$ received a potential of $1/\sigma_{x,t}^2$ from operation $t$: by definition, $\sigma_{x,t}$ counts the items in the heap at time $t$ that will eventually be deleted and whose insertion time is greater or equal to that of $x$. Those are exactly the items that will get potentials $1/1^2, 1/2^2, \ldots, 1/\sigma_{x,t}^2$, with $x$ getting the last term.

	Using that the strict working set of $x$ is $\strictw_x = \max_{t_x < t' \le t'_x} \sigma_{x,t'} \ge \sigma_{x,t}$, we conclude that $x$ gets at least $1/\strictw_x^2$ potential from every $t$. Thus, the total potential sent to $x$ is at least $(t'_x - t_x)/\strictw_x^2$. Using that $x$ is unhappy, we conclude that this expression is at least $(t'_x - t_x) / (\sqrt[3]{t'_x - t_x})^2 = \sqrt[3]{t'_x - t_x} = \Omega(\log (t'_x - t_x))$ which concludes the proof. 
\end{proof}

Next we prove the remaining equivalence, between the delete-only working-set property and the regular working-set property.

We note that the proof of \cref{thm:delete_implies_working_set} is very similar to that of \cref{thm:weak_strong_equivalence}. The difference is that previously, we were determining which items should get how much potential by their insertion order, while now we will do so using their deletion order.

\begin{theorem}
    \label{thm:delete_implies_working_set}
	Any heap with the working-set property can be shown to also have the delete-only working-set property, in the amortized sense, with only an additive $\O(1)$ amortized overhead to all operations.
\end{theorem}
\begin{proof}
	The proof closely follows that of \cref{thm:weak_strong_equivalence}. This time, we want to show an amortized $\Delete$ cost of $\O(\log \delta_x)$. We consider $x$ \emph{happy} if $t'_x - t_x \le \delta_x^3$. We carry on with the proof until the description of how the operation at time $t$ redistributes its potential.

	We consider all items in the heap at time $t$ that will eventually be deleted. Denote them $x_1, x_2, \ldots$, ordered this time in increasing order by their deletion time. Then we send a potential $1/i^2$ to each $x_i$.

Now, consider any unhappy item $x$, and consider the operation at time $t_x < t \le t_x'$. By the definition of $\delta_x$, there are at most $\delta_x$ deletions between $t$ and $t_x'$, and thus $x$ receives a potential of $\ge 1/\delta_x^2$ from the operation at $t$. By a similar argument as in the proof of \cref{thm:weak_strong_equivalence}, we conclude that $x$ gets at least $(t'_x - t_x) / \sqrt[3]{t'_x - t_x}^2 = \Omega(\log(t'_x - t_x))$ potential, which finishes the proof.
\end{proof}

We are now ready to prove the theorem from \cref{s:intro}:

\allpropertiessame*
\begin{proof}
	\cref{rem:touch_better_than_regular,rem:insert_better_than_touch,rem:strong_better_than_insert,rem:strict_better_than_strong} show the implications $1 \leftarrow 2$, $2\leftarrow 3$, $2\leftarrow 4$, $3\leftarrow 5$, $5\leftarrow 6$. \cref{thm:weak_strong_equivalence} shows $6\leftarrow 1$ and \cref{thm:delete_implies_working_set} shows $4 \leftarrow 1$.
\end{proof}

\subsection{The stack-like property}
\label{sec:stack_like}

In this section, we define the stack-like property and prove that it is strictly stronger than the working-set property.

Recall that the \emph{stack size} $s_x$ of $x$ is defined as the number of items inserted at or after $t_x$ that are still present during the deletion of $x$.

\begin{definition}[Stack-like property]
    \label{def:stack-like}
            We say a heap has the \emph{stack-like} property, if for every sequence of requests, the amortized complexity of $\ExtractMin$ is $\O(\log s_x)$.
\end{definition}

We are ready to prove \cref{lem:stacklike_stronger}:

\stacklikestronger*

\begin{proof}
	\begin{enumerate}
		\item Thanks to \cref{thm:all_same}, it is enough to show that the stack-like property is always at least as good as the insert-only working-set property. This holds, since $s_x$ counts the subset of all items inserted between $t_x$ and $t'_x$ and the insert-only property counts all of them.
		\item 
This can be seen from the input 
\begin{align}
\label{eq:stack-good}
I(n), I(n-1), \dots, I(1), \underbrace{E, E, \dots, E}_{\text{$n$ times}} 
\end{align}
where a heap with the working-set property pays the overall cost of $\O(n \log n)$, while any heap with the stack-like property pays only $\O(n)$. \qedhere
	\end{enumerate}
\end{proof}

\subsection{The size property}
\label{sec:size_property}

For completeness, we turn our attention to what we call the \emph{size property}, a relatively weak beyond-worst-case property that is satisfied by many (if not all) well-known heap implementations.
It says that the cost of \Deletemin{} should be logarithmic in the current size of the heap (rather than, say, in the total number of operations, or the largest-ever size, etc.).

\begin{definition}[Size property]
    \label{def:size}
	A heap has the \emph{size property}, if for every sequence of requests, the amortized complexity of $\ExtractMin$ is $\O(\log |H(t'_x)|)$.
\end{definition}

\paragraph{The working-set property implies the size property}
First, we prove that the working-set property implies the size property. The idea of the proof loosely follows that of \cref{thm:weak_strong_equivalence}, although the potential is redistributed in a different way.

\begin{theorem}
    \label{thm:working_set_implies_size}
	Any heap with the working-set property can be shown to also have the size property, in the amortized sense, with only an additive $\O(1)$ overhead to all operations.
\end{theorem}
\begin{proof}
Fix any sequence $\rho$ of requests \Insert, \Touch, \Delete. Also, fix a heap with the working-set property. Our goal is to set up an amortization argument under which each $\Delete$ cost is $\O(\log |H(t'_x)|)$ and other costs only increase by an additive $\O(1)$.

    We classify each inserted item $x$ as either \emph{happy} or \emph{unhappy}; it is happy if $t'_x - t_x \le |H(t'_x)|^4$. As in \cref{thm:weak_strong_equivalence}, happy items are paid for, and we only need to show how we can pay for the unhappy items. That is, we want to redistribute the $\O(1)$ potential given to each operation in $\rho$ such that every unhappy $x$ ends up with $\Omega(\log(t'_x - t_x))$ potential.

    Consider the operation that happened at time $t$, and consider the set of items whose $\Delete$ cost is affected by this operation: all items in the heap at time $t$ that will eventually be deleted. To each such item $y$, we send a potential of $1/|H(t'_y)|^3$. Note that for every $i$, there are at most $i$ items $y$ that receive exactly $1/i^3$ potential from operation $t$. That is because we can, among all such items, pick the $y$ extracted the earliest and note that since no other $z$ with $|H(t'_z)| = i$ has been extracted yet at $t'_y$, we must have $z \in H(t'_y)$ for all of them. This means that in total, the operation at time $t$ sends out at most $1 \cdot 1/1^3 + 2 \cdot 1/2^3 + 3 \cdot 1/3^3 + \ldots = \O(1)$ potential.

	Finally, we need to verify that each $\Delete$ of an unhappy $x$ has received enough potential. An item $x$ received $(t'_x - t_x) / |H(t'_x)|^3$ potential. By the assumption of unhappiness, this is at least $(t'_x - t_x) / \sqrt[4]{(t'_x - t_x)^3} = \Omega(\log(t'_x - t_x))$, which concludes the proof.
\end{proof}

\paragraph{Working set is asymptotically stronger than the size property.}
To see this, consider a sequence 
\begin{align}
\label{eq:sponge}
I(1), I(2), \dots, I(n_1), \underbrace{I(0), E, I(0), E, \dots }_{\text{$n_2$ times}}
\end{align}

In this sequence, a heap with the size property guarantees that the total time of all operations is $\O(n_1 + n_2 \cdot \log n_1)$. On the other hand, a heap with the working-set property finishes in time $\O(n_1 + n_2)$.

\section{Simple heaps with a working-set bound}
\label{sec:simple-heaps}

Here we review some simple data structures that have a working-set bound, provided there are no $\DecreaseKey$ operations.

In designing a heap implementation with a working-set bound, there are two item orders that are relevant: key order and insertion order.  There is a standard data structure that naturally supports two node orders: a binary search tree.  In such a tree, the nodes are totally ordered in symmetric order and partially ordered in heap order (parent less than child).  It is natural to use such a tree to represent a heap.  One way is to form a binary tree in which the nodes are the heap items, in symmetric order by decreasing insertion time: The leftmost node is the most recently inserted, the rightmost node is the least recently inserted.  If $x$ is the $k^{\rm{th}}$ node in increasing symmetric order, there are $k$ items in its working set that are currently in the heap.  To facilitate $\Deletemin$, we store in each node a pointer to a node of minimum key in its subtree.  If the tree is a balanced tree, such as a weak AVL~\cite{weak-avl} or red-black tree~\cite{redblack}, the worst-case time for a $\FindMin$ is $\OO(1)$, and that of an $\Insert$ or $\Deletemin$ is $\OO(\log n)$.

We want a working-set bound for $\Deletemin$.  There are at least two ways to obtain such a bound for this data structure.  One is to use a splay tree~\cite{splay} instead of a balanced tree.  The other is to modify the tree representation to support fast accesses to recently inserted items; that is, to use a finger search tree.  We describe both of these solutions.

\subsection{A splay heap has a working-set bound}

 We shall show that a heap implemented as a splay tree as described above has the insert-only working-set bound (and thus all listed working-set bounds), provided that it ends empty.  Our analysis closely resembles Sleator and Tarjan's proof that splay trees have a search-tree working-set bound.  We use the splay-tree access lemma~\cite{splay}.  This lemma states that if each node $x$ of a splay tree is assigned an arbitrary positive individual weight $w(x)$, the amortized time to access a node $x$ and move it to the root via a splay operation is $\OO(1+\log(W/w(x)))$, where $W$ is the sum of the individual weights of all the nodes.  The proof of this result defines the potential of a node to be $\log(W(x))$, where $W(x)$ is the sum of the weights of the nodes in the subtree rooted at $x$. The potential of a tree is the sum of the potentials of its nodes.

To use this result to obtain a heap working-set bound, we define the individual weight of a node $x$ in the tree to be $1/k^2$, where $k$ is the number of items inserted into the tree from the time $x$ was inserted until the present. With this assignment of weights, the individual weight of a node is a non-increasing function of time; node weights are in decreasing order by insertion time; when an item is deleted, its individual weight is the square of the reciprocal of its working-set size; and $W$, the sum of the individual node weights of the nodes in the tree, is always at most $\sum_{i=1}^\infty 1/i^2$, a fixed constant.

To do an $\Insert$, we make the new node the root of the tree, and we make the old root its right child.  This takes $\OO(1)$ time.  The new root has at most constant potential; the potentials of all other nodes decrease.  Thus adding the new node increases the potential of the tree by at most constant.  It follows that the amortized time for an insert is $\OO(1)$.

To do a $\Deletemin$, we find a node of minimum key, say $x$, by following the pointer to such a node from the root. If $x$ has two children, we find its predecessor in symmetric order and swap $x$ with its predecessor.  Now $x$ has at most one child.  We delete $x$ from the tree, replace it by its only child if it had one, and splay its old parent. The time for the $\Deletemin$ is at most a constant times the time for the splay. To bound the amortized time of the $\Deletemin$, we observe that swapping $x$ and its predecessor decreases the potential of the tree, since nodes in symmetric order are in decreasing order by weight.  To bound the amortized time of the splay, when we delete $x$ temporarily add its individual weight to that of its parent, the node being splayed.  This does not increase the potential of the tree.  With this change, the access lemma implies that the amortized time for the splay is $\OO(\log i_x)$, where $i_x$ is the insert-only working-set size of $x$.  At the end of the splay, we subtract the extra individual weight from the splayed node, which decreases the potential of the tree.

We conclude that a splay heap has an amortized working-set bound, provided that the heap ends empty.  The reason  the heap must end empty for our proof to work is that with our definition of the potential, a non-empty tree can have negative potential, which the analysis does not take into account.

\subsection{Finger heaps are stack-like}

We obtain a stronger bound by implementing a heap as a balanced tree, but changing the tree representation.  We make the tree into a finger tree, with a finger at the leftmost node in the tree.  We accomplish this by reversing the left child pointers along the left spine of the tree (the path from the root through left children to the smallest node in symmetric order).  Access is via the smallest node in symmetric order.  Each node not on the left spine has pointers to its left and right children; each node on the left spine has pointers to its parent and to its right child.  In addition, each node not on the left spine has a pointer to a node of minimum key in its subtree; each node on the left spine has a pointer to a node of minimum key among those larger than it in symmetric order, which are those nodes reachable from it by following pointers.

If the tree is a kind that supports rebalancing after an insertion or deletion in $\OO(1)$ amortized time, such as a red-black or weak AVL tree, the resulting heap implementation is stack-like: The amortized time for a $\Deletemin$ is logarithmic in the stack size of the deleted item, $\FindMin$ takes $\OO(1)$ time worst-case, and $\Insert$ takes $\OO(1)$ amortized time.  The bounds can be made worst-case by maintaining a regularity condition on the left spine, using the same idea as in a redundant digital numbering scheme~\cite{digital-numbering1,brodal,digital-numbering2}.

If we add a finger to the rightmost node in the tree, by modifying the right spine in the same way as the left spine, we can delete the $k^{\rm{th}}$ least recently inserted item or the $k^{\rm{th}}$ most recently inserted item in $\OO(1+\log k)$ time.  This gives us a heap that is worst-case \emph{deque-like}.

This is the simplest implementation of a deque-like (or stack-like, or queue-like) heap that we know.

\ifanonymous\else
\paragraph{Acknowledgments}

BH, RH, and VR were partially funded by the Ministry of Education and Science of Bulgaria's support for INSAIT as part of the Bulgarian National Roadmap for Research Infrastructure. 
BH and RH were partially funded through the European Research Council (ERC) under the European Union's Horizon Europe research and innovation program (ERC Advanced grant agreement 101268062 and ERC Starting grant agreement 949272).
VR was partially funded through the European Research Council (ERC) under the European Union's Horizon 2020 research and innovation program (ERC grant agreement 853109).
VR was partially supported by Charles Univ. project UNCE 24/SCI/008 and by the Czech Science Foundation (GAČR), grant no. 26-23599M.
RH was supported by the VILLUM Foundation grant 54451. 
Part of this work was done while RH was visiting BARC at the University of Copenhagen. RH would like to thank Rasmus Pagh for hosting him there. Part of this work was done while RH was visiting INSAIT.
Part of this work was done while VR was a postdoc at INSAIT.
RT's research at Princeton was partially supported by a gift from Microsoft.  Part of this work was done during RT's visits to INSAIT and to the Simons Institute for the Theory of Computing.
\fi

\printbibliography



\end{document}